\documentclass[10pt, conference, letterpaper]{ IEEEtran}
\IEEEoverridecommandlockouts

\usepackage{amsmath,amssymb,amsfonts}
\usepackage{amsthm}
\usepackage{mathrsfs}
\usepackage{stmaryrd}
\usepackage{algorithm}
\usepackage{multicol}
\usepackage[export]{adjustbox}
\usepackage{textcomp}
\usepackage[dvipsnames]{xcolor}
\usepackage{booktabs}
\usepackage{subfigure}
\usepackage{comment}
\usepackage[commentColor=Orange!70!Gray,indLines=false,rightComments=false]{algpseudocodex}
\usepackage{amsthm}
\usepackage{array}
\usepackage{bbding}

\usepackage{enumitem}
\usepackage[T1]{fontenc}
\usepackage{calc}
\newtheorem{theorem}{Theorem}
\newtheorem{lemma}{Lemma}

\newtheorem{proposition}{Proposition}

\def\BibTeX{{\rm B\kern-.05em{\sc i\kern-.025em b}\kern-.08em
    T\kern-.1667em\lower.7ex\hbox{E}\kern-.125emX}}
\begin{document}

\title{Reputation-Based Leader Election under Partial Synchrony: Towards a Protocol-Independent Abstraction with Enhanced Guarantees\\
}

\author{Xuyang Liu$^{\dagger,\ddagger}$, Zijian Zhang$^{\dagger,*}$, Zhen Li$^\dagger$, Jiahang Sun$^\S$, Jiamou Liu$^\ddagger$, Peng Jiang$^{\dagger,*}$\\
	$^\dagger$School of Cyberspace Science and Technology, Beijing Institute of Technology, China\\
	$^\ddagger$School of Computer Science, The University of Auckland, New Zealand\\
	$^\S$School of Computer Science, Beijing Institute of Technology, China \\
	Email: \{liuxuyang, zhangzijian, zhen.li, sunjh, pengjiang\}@bit.edu.cn, jiamou.liu@auckland.ac.nz\\ 
	$^*$Corresponding author
}

\maketitle

\begin{abstract}
Leader election serves a well-defined role in leader-based Byzantine Fault Tolerant (BFT) protocols. Existing reputation-based leader election frameworks for partially synchronous BFTs suffer from either protocol-specific proofs, narrow applicability, or unbounded recovery after network stabilization, leaving an open problem. This paper presents a novel protocol-independent abstraction formalizing generic correctness properties and effectiveness guarantees for leader election under partial synchrony, enabling protocol-independent analysis and design. Building on this, we design the Sliding Window Leader Election (SWLE) mechanism. SWLE dynamically adjusts leader nominations via consensus-behavior-based reputation scores, enforcing Byzantine-cost amplification. We demonstrate SWLE introduces minimal extra overhead to the base protocol and prove it satisfies all abstraction properties and provides superior effectiveness. We show, with a 16-server deployment across 4 different regions in northern China, SWLE achieves up to 4.2× higher throughput, 75\% lower latency and 27\% Byzantine leader frequency compared to the state-of-the-art solution under common Byzantine faults, while maintaining efficiency in fault-free scenarios.

\end{abstract}


\section{Introduction}
\noindent \textbf{Overview.} 
This paper addresses limitations in existing reputation-based leader election frameworks for partially synchronous BFTs that either tie correctness to specific protocols, target narrow scenarios or protocol classes, or lack bounded recovery post-GST.  We present a novel protocol-independent abstraction formalizing generic properties for theoretical analysis. Based on this, we design a novel mechanism, providing proven correctness, effectiveness (including bounded post-GST recovery), significantly outperforming state-of-the-art approach under common Byzantine faults.

\noindent \textbf{History.} Byzantine Fault Tolerant (BFT) state machine replication originated from Lamport et al.'s Byzantine Generals Problem \cite{lamport2019byzantine}, which formalized achieving deterministic consensus in distributed systems with malicious actors. They proved the impossibility of solving consensus with $f$ traitors among $3f$ generals, while providing a solution for $3f+1$ replicas under synchrony with high communication complexity. Fischer et al. \cite{fischer1985impossibility} later showed deterministic consensus is impossible in asynchronous settings even with a single fault, necessitating randomness for liveness. Ben-Or \cite{ben1983another} established the fundamental $(n-1)/3$ resilience bound under asynchrony, requiring at least $3f+1$ replicas to tolerate $f$ Byzantine faults.

\noindent \textbf{Partial Synchrony.} In 1988, Dwork et al. proposed the partially synchronous model \cite{dwork1988consensus}, introducing the Global Stabilization Time (GST) - an unknown finite time after which network becomes synchronous. This model enables BFT protocols to ensure safety during all periods and liveness after GST, provided new directions for designing BFT protocols.

\noindent \textbf{Leader-Based Paradigm.} Leader-based paradigms emerged as dominant approach in consensus development: one replica serves as leader/primary to coordinate consensus while others function as voters. Initially introduced for conflict resolution to ensure liveness by providing a single authority to order conflicting operations \cite{lamport2006fast}, the role evolved to leaders being fully responsible for initiating proposals. The typical consensus flow involves leader election, proposal dissemination, voting, and decisions based on supermajority (mainly BFT) or majority (mainly Crash Fault Tolerence (CFT)) agreement. If consensus is achieved, system commits and progresses; otherwise, it aborts and elects a new leader to restart the process.

\noindent \textbf{Unique Benefits of Partially Synchronous \& Leader-Based BFTs.} Partial synchrony offers unique advantages over full synchrony or asynchrony, occupying an intermediate position—closely modelling real-world networks with bounded delays and inherent uncertainties. It organizes consensus into views \cite{castro1999practical}, providing an elegant protocol design framework that balances safety and liveness guarantees. The leader-based paradigm further amplifies these benefits through coordinated leadership, delivering reduced communication complexity, simplified protocol design, and increased efficiency. Consequently, these protocols become enormously impactful and form core infrastructure of permissioned blockchains \cite{androulaki2018hyperledger}.

\noindent \textbf{Leader Election under Partial Synchrony: Challenges, Existing Solutions, and the Open Problem.}  
Classical BFTs \cite{castro1999practical,gueta2019sbft} employ stable leaders until failure which is then replaced by designated successors. While blockchain-inspired chain-based BFTs \cite{yin2019hotstuff,abraham2020sync,liu2024dolphin} introduce leader rotation to address chain quality in trustless environments where replicas should fairly share responsibilities. Under leader rotation, replicas alternate as leaders following predefined rules, ensuring correct replicas lead infinitely with fixed proportions as view increases, but unfortunately extending same guarantee to Byzantine replicas.
Since leaders control proposals, Byzantine leaders can severely disrupt consensus. Most fundamentally, they can simply drop messages until timeout-triggered rotation occurs (cost-free), causing considerable performance degradation as timeouts typically far exceed normal commit latencies.

Several approaches have emerged to mitigate Byzantine leader impact, targeting transaction ordering \cite{kelkar2020order,kelkar2023themis}, communication protocols \cite{alhaddad2022balanced,yang2022dispersedledger}, or imposing computational costs through hash-based elections \cite{zhang2021prosecutor}. Recent research has explored reputation-based election approaches \cite{zhang2024prestigebft, liu2025abse,cohen2022aware,tsimos2024hammerhead} where reputation-based signifies relying solely on available information during consensus process to determine leaders, offering self-contained solutions that avoid external dependencies.
However, state-of-the-art solutions primarily focus on crash-only scenarios with limited practical significance, and none simultaneously achieve bounded recovery after GST, protocol-independence and general applicability to partially synchronous BFTs, leaving this as an open problem.

\subsection{Our Contributions.}

To address these limitations, this paper presents a comprehensive solution spanning theoretical foundations, mechanism design, and practical implementation. 

\vspace{1mm}
\noindent \textbf{Theoretical foundations.} We first establish a generic abstraction framework, defining three core properties that captures the essential correctness and effectiveness guarantees of leader election mechanisms in partially synchronous settings. This decouples correctness proofs from specific protocols, enabling protocol-independent analysis and design. 

\vspace{1mm}
\noindent \textbf{Mechanism design.} Building upon this foundation, we design the Sliding Window Leader Election (SWLE) mechanism that dynamically adjusts leader nominations based on replica behavior. The core design philosophy of SWLE centers on:
\begin{itemize}[leftmargin=13pt]
	\item Determining Leadership eligibility by accumulated reputation reflecting consensus participation and behavior quality.
	\item Separating election from current view by finalizing future leaders (sufficiently in advance to maintain liveness post-GST) through authenticated consensus decisions, eliminating race conditions between election and view progression.
	\item Maintaining a fixed-size window for future leaders and triggering fair, systematic elections universally.
	\item Enforcing Byzantine-cost amplification and introducing minimal extra overhead to the base protocol.
\end{itemize}

\noindent It integrates with any BFT satisfying two prerequisites (align with core properties of most partially synchronous BFTs). We prove SWLE satisfies our abstraction and provides superior effectiveness compared to conventional approaches while maintaining bounded recovery after network stabilization.

\vspace{1mm}
\noindent \textbf{Practical Implementation.} SWLE is implemented as a lightweight Rust module ($\sim$400 LoC) and its integration requires minimal modifications (e.g., +$\sim$300 LoC in non-chained HotStuff \cite{yin2019hotstuff}). We show, with a 16-server deployment across 4  regions in northern China, SWLE achieves up to 4.2× higher throughput, 75\% lower latency and 27\% Byzantine leader frequency under common Byzantine faults compared to state-of-the-art solution ABSE \cite{liu2025abse} and Hotstuff's original  election mechanism, ensures bounded recovery post-GST, while maintaining minor extra overhead in fault-free scenarios.

\textbf{Thus, this paper makes the following contributions:}

\begin{itemize}[leftmargin=13pt]
	\item A protocol-independent abstraction formalizing correctness and effectiveness for reputation-based leader election.
	\item SWLE mechanism build upon the abstaction, with formal proofs of satisfying all properties under partial synchrony.
	\item An open-source, lightweight implementation demonstrating seamless integration and superior performance.
\end{itemize}

\section{Related Work}

\noindent \textbf{Fault Tolerant Consensus.} Fault-tolerant consensus in distributed systems can be primarily categorized into CFT and BFT, with CFT addressing benign failures. Most CFTs stem from Paxos \cite{lamport2001paxos}, which influenced many subsequent designs, including BFTs. BFT addresses arbitrary failures and can be categorized into synchronous \cite{feldman1997optimal,abraham2020sync}, partially synchronous \cite{castro1999practical,yin2019hotstuff,liu2025partially}, and asynchronous \cite{gao2022dumbo,duan2023fin} variants based on network assumptions. In this work, we focus on partially synchronous BFTs, widely adopted as core building block for permissioned blockchains \cite{androulaki2018hyperledger}.

\vspace{1mm}
\noindent \textbf{Leader in BFTs.} Most BFTs employ a leader-based model. While leader roles vary across protocols \cite{castro1999practical,liu2025group,spiegelman2022bullshark}, they fundamentally coordinate ordering and facilitate agreement among replicas. Several studies \cite{borran2010leader,aublin2013rbft} identified leaders' significant influence on consensus performance, with some exploring leaderless alternatives \cite{miller2016honey,crain2021red}.
In this work, we focus on enhancing leader-based protocols through reputation mechanisms while preserving coordination advantages. 

\vspace{1mm}
\noindent \textbf{Chain-based BFTs.} The application of BFTs to blockchains spurred innovations like chain-based BFTs \cite{yin2019hotstuff,abraham2020sync,liu2024dolphin}, which introduce leader rotation to address chain quality, and employ a linear block structure to streamline message patterns and enable decision pipelining. A significant advancement is HotStuff \cite{yin2019hotstuff}, which first develops a basic non-chained version by combining features from SBFT \cite{gueta2019sbft} (collective signatures) and Tendermint \cite{buchman2016tendermint} (block locking), then extends to a chained variant by incorporating pipelining techniques. While chain-based BFTs share core safety with classical designs, liveness exhibit subtle distinctions under partial synchrony. Due to space constraints, our work establishes SWLE’s properties under classic pattern to maintain clarity of our fundamental approach, but note it can be extended seamlessly to chain-based variants, similar to extending BFT's properties from classic pattern to chain-based pattern.

\vspace{1mm}
\noindent \textbf{Existing Approaches Mitigating Byzantine Leader Impact.} A wealth of research targets leader vulnerabilities in BFTs. Approaches like Aequitas \cite{kelkar2020order} and Themis \cite{kelkar2023themis} mitigate adversarial ordering by guaranteeing transaction order-fairness. Prosecutor \cite{zhang2021prosecutor} hinders Byzantine servers from attaining leadership by imposing computational burdens via hash-based election campaigns. Communication-level solutions \cite{alhaddad2022balanced,yang2022dispersedledger} leverage balanced reliable broadcast or verifiable information dispersal to eliminate leaders as communication hubs during broadcasts—though often at excessive overhead. 

Recent reputation-based approaches offer self-contained alternatives. PrestigeBFT \cite{zhang2024prestigebft} uses reputation-based computational penalties for active view changes. ABSE \cite{liu2025abse} introduces an adaptive baseline score-based election framework but lacks protocol-independence and bounded post-GST recovery. Carousel \cite{cohen2022aware} uses on-chain data for local leader determination without consensus dependency, but focuses primarily on crash-only scenarios in chain-based BFTs; While Tsimos et al. \cite{tsimos2024hammerhead} extend its concepts to DAG-based BFTs, it creates framework divergence rather than unification. We provide detailed comparisons between our abstraction and state-of-the-art frameworks in Section \ref{subsec:comparison}.

\section{System Model}

We consider a state machine replication (SMR) \cite{lamport2019time} system consisting of $n$ replicas denoted as $\mathcal{R} = \{1, 2, \ldots, n\}$. 
We assume a Public Key Infrastructure (PKI) as the cryptographic primitive where each replica $i$ has a fixed public-private key pair $(pk_i, sk_i)$ during protocol execution. Replicas use digital signatures to authenticate messages, which are secure against computationally bounded adversaries. The system operates under the partial synchrony network model \cite{dwork1988consensus}: there exists an unknown Global Stabilization Time (GST) such that after GST, all messages between correct replicas are delivered within a known bounded delay $\Delta$. Before GST, message delays are unbounded, though communication channels remain reliable (ensuring eventual message delivery) and point-to-point.

We consider a Byzantine adversary controlling up to $f < n/3$ replicas (trusted hardware-dependent \cite{niu2025achilles} protocols may achieve a resilience bound of $(n-1)/2$; for simplicity, we omit this here and define $n = 3f + 1$), which may behave arbitrarily: they can collude, produce arbitrary values, eavesdrop, selectively delay or omit messages, and send conflicting information, but can not compromise the cryptographic primitives, prevent point-to-point communications between correct replicas, or indefinitely delay the occurrence of GST. Correct replicas strictly follow the protocol specification.

\section{Leader Election Abstraction}
\label{sec:abstraction}

Leader election serves a well-defined role in leader-based BFT protocols. In this work, we provide a generic foundation for designing reputation-based leader election mechanisms under partial synchrony, where reputation-based signifies leader determination relies solely on information available during consensus process or consensus-generated on-chain information. Instead of protocol-specific designs, we emphasize abstracting election patterns prevalent in most partially synchronous protocols. We then demonstrate how it enables design of a novel mechanism to replace conventional approaches.

\subsection{Leader-based BFTs under Partial Synchrony}
\label{subsec:structure}

Generally speaking, any leader-based BFT includes at least a  round-based normal state protocol. Partially synchronous BFTs additionally rely on view changes (with a form of leader rotations) for liveness and therefore proceed in views. For self-containment, before abstracting properties, we provide a brief high-level description of the workflow of such protocols:

The system progresses through sequential views $v = 1, 2, \dots$, where during $v$ a replica in $\mathcal{R}$, chosen through an election process, serves as the leader to drive the consensus protocol for one consensus round (explicitly assigned a unique sequence number in some protocols \cite{castro1999practical}) or more (some protocols separate data transition from consensus \cite{spiegelman2022bullshark} or create multiple instances \cite{lyu2025ladon}; for such protocols, we consider only one instance by default). Each consensus round comprises one or more phases, which can be broadly categorized into:

\begin{enumerate}
	\item \textbf{Dissemination phases}: where the leader broadcasts information (proposals/updates) to all replicas. For brevity, we will refer to these collectively as proposals.
	\item \textbf{Voting phases}: where replicas vote on the proposals, with most protocols requiring a supermajority (i.e., $2f+1$ votes from distinct replicas) for acceptance.
\end{enumerate}

A view change is triggered when a leader appears to be faulty (i.e., timer expiration). Some protocols (without affecting timeout-based changes) also include auxiliary strategies (e.g., triggered periodically \cite{veronese2009spin}), or have view change mechanisms embedded within the normal state protocol, causing leaders to rotate with consensus rounds  \cite{yin2019hotstuff}. The two most common leader election strategies include rolling-based selection \cite{liu2024dolphin,yin2019hotstuff} (round-robin fashion, e.g., the leader of $v$ is replica $v \bmod n$) and predefined leaders for each view before consensus begins \cite{castro1999practical,spiegelman2022bullshark}. 

Asynchronous BFTs, however, do not have view change mechanisms. Furthermore, the FLP result \cite{fischer1985impossibility} demonstrates that the fault tolerance problem cannot be solved deterministically in asynchronous settings. Therefore, additional randomness is required during leader election to guarantee liveness, typically relying on additional building blocks (e.g., common coins \cite{sui2025signature}). Due to such significant differences, we do not jointly consider asynchronous assumptions, which could result in properties that are either too broad with model-specific redundancies, or too complex for convenient mechanism design. We aim for properties that are concise, intuitive, and can effectively guide the design of new mechanisms.

\subsection{Property Specification}
\label{subsec:prospe} 
In our formalization, we abstract generic correctness properties for leader election mechanisms under partial synchrony, some of which are derived from existing protocol designs. Our core objective is to enable a single universal correctness proof for any leader election mechanism satisfying these properties, eliminating the need for protocol-specific proofs. 
Let $\textsc{leader}_j(v)$ denote replica $j$'s determination of the elected leader for view $v$. We first define the following fundamental properties:

\begin{description}
	\item[Leader Uniqueness:] \label{prop:uniqueness}
	For any view $v$, there do not exist two distinct replicas $j$ and $k$ such that both can gather a quorum of votes (i.e., $2f+1$ distinct votes) for conflicting leadership claims in $v$.
	\item[Timely Finalization:] \label{prop:timeliness}
	For any view $v$, every correct replica $j$ finalizes $\textsc{leader}_j(v)$ before or upon entering $v$.
\end{description}

Leader Uniqueness is a safety property, while Timely Finalization is a liveness property (dependent on eventual synchrony after GST). Briefly, common leader election mechanisms under partial synchrony essentially fix the leader per view before consensus begins, hence easily satisfying both properties. This is, however, not the case for reputation-based leader election mechanisms, because replica reputation needs to be established during protocol execution -- leader selection strategies needs to change dynamically as consensus progresses. Nevertheless, it is also important to note a leader election mechanism should focus solely on the process of electing leaders and \textbf{should not} alter the leader's functional role in the protocol (e.g., Byzantine leader tolerance mechanisms should remain unchanged. The key enhancement should lie in reducing the probability of Byzantine leaders being elected through reputation mechanisms, thereby decreasing the frequency of scenarios where "tolerating Byzantine leaders" becomes necessary). 

While the above two properties ensure basic consistency and leader availability, they are still insufficient for designing a effective leader election mechanism. Consider the worst-case scenarios for a mechanism satisfying both properties:
 \begin{enumerate}
 	\item The mechanism consistently elect Byzantine leaders.
 	\item Correct replicas finalize different leaders for same views (i.e. no replicas can receive $2f+1$ votes during elections).
 \end{enumerate}
In both cases, consensus cannot progress normally (liveness issues). This motivates an extra property to quantify a mechanism's effectiveness - measuring how well it promotes progress by electing unified correct leaders among correct replicas. This can also be used to theoretically evaluate whether it can serve as a replacement for a protocol's original mechanism
\vspace{1mm}

\noindent \textbf{Attempt 1.} An ideal enhancement to a leader election mechanism that could implement eventual perfect failure detection may be defined as: faulty replicas are permanently identified and correct replicas are never suspected (or, equivalently, for any correct replica there always exist views where it is elected among correct replicas, and no view exists where a Byzantine replica is elected). Unfortunately, a reputation-based leader election mechanism cannot achieve this definition because:

\begin{itemize}
	\item It must allow suspected correct replicas to be unsuspected based on good behavior and hence cannot permanently isolate a Byzantine replica once the latter starts behaving well and accumulates sufficient reputation.
	\item Byzantine replicas can choose to behave correctly, making them indistinguishable from correct replicas.
	\item Under network delays or asynchronous periods, even correct replicas might occasionally fail to accumulate reputation, thus being suspected.
\end{itemize}

\noindent \textbf{Attempt 2.} Another choice is to define the following property:

\begin{description}
	\item[Increased Adversarial Cost:] \label{prop:increased}
	Byzantine replicas incur higher cost of misbehavior to be elected as leaders.
\end{description}

This property essentially creates a tradeoff, which in turn indirectly reflects the effectiveness of a leader election mechanism, and is more feasible than Attempt 1: Byzantine replicas can no longer simply wait for their turn in the rotation without any good behavior required to eventually become leaders. For example, if Byzantine replicas need to contribute positively to consensus for a substantial period before gaining leadership eligibility, the overall system performance improves compared to scenarios where they can disrupt without such constraints.

Unfortunately, first, this property remains incomplete because it still does not consider whether correct replicas can be elected. Second, terms like "higher" are essentially comparative properties that require reference to a baseline mechanism (e.g., protocol's original mechanism), which violates our protocol-independent abstraction goal. Furthermore, such terms are too vague and cannot well measure effectiveness.

We comment that while we do not adopt Increased Adversarial Cost, we retain its core intuition. A leader election mechanism can indeed be motivated by this idea. Moreover, if a mechanism (whether reputation-based or not) preserves the original leadership eligibility (under common mechanisms) of correct replicas, then the property will naturally apply to it.

\noindent \textbf{Final Attempt.} Finally, we define the following property:

\begin{description}
	\item[$\gamma$-Guarantee:] A leader election mechanism satisfies the $\gamma$-Guarantee (where $\gamma > 0$) if there exist positive integers $v_c$ and $T$, and a constant $\sup \leq T$, such that for any view $v > v_c$, over any sequence of $T$ consecutive views starting from $v$, the expected number of views $E[C]$ where:
	(1) All correct replicas agree on the same elected leader, and
	(2) This leader is correct,
	satisfies:$\sup \geq E[C] \geq \gamma T$
\end{description}

$\gamma$-Guarantee is also a liveness property, and is slightly unconventional. It is somewhat similar to the chain quality property in Bitcoin \cite{garay2024bitcoin}, having both deterministic guarantees (addressing the liveness issues) and measurement indicators (effectiveness). We explain why $E[C]$ is bounded by a range:

\vspace{1mm}

\noindent \textit{Lower Bound ($\gamma T$)}: Acknowledges that Byzantine replica behavior is variable and that perfect detection is impossible (Attempt 1). When Byzantine replicas behave correctly, they become indistinguishable from correct replicas under reputation-based leader election mechanisms and cannot be identified. It is hoped that this situation serves as the lower bound for a leader election mechanism: Byzantine replicas cannot further reduce the effectiveness guarantee (i.e. reduce $E[C]$ below $\gamma T$) through collusion or manipulation. It is also hoped this compels Byzantine replicas to maintain "good behavior" to the greatest extent possible\footnote{We abandon the quantification of Byzantine replica misbehavior costs to avoid referencing other mechanisms in the property (avoiding protocol-specific tendencies), i.e., it can only serve as an informal indicator. Similarly to the aforementioned, it is also suitable as motivation for new mechanism design.} (similar to correct replicas) to remain eligible for leadership - indirectly increasing their cost of misbehavior as discussed in Attempt 2, which has a positive impact on system overall performance.

\vspace{1mm}

\noindent \textit{Upper Bound ($\sup$)}: The bound represents the best-case effectiveness under optimal conditions. Different mechanisms may implement different scenarios when achieving it (e.g., consider the case where all faults are crash faults -- a common approach, although in practice, it has limited effect on a BFT system).

We now explain why we consider correct replicas collectively rather than individually: Correct replicas may exhibit heterogeneous capabilities (e.g., network latency). Mechanisms should permit leadership concentration on higher-performing correct replicas (e.g., those with lower average message delay) to enhance system throughput.

Leader Uniqueness combined with Timely Finalization and $\gamma$-Guarantee guarantees the correctness of a leader election mechanism, while its effectiveness needs to be measured according to parameters (theoretically). A mechanism can be correct ($\gamma$ is small) but ineffective. Standard protocols with common mechanisms (e.g., round-robin with $3f+1$ replicas) typically achieve: $(v_c = 0, T = n, \gamma = \frac{2f+1}{n}, \sup = \gamma T)$.  Reputation-based mechanisms usually require $v_c > 0$ for reputation initialization, but target a higher $\sup$ or $\gamma$.

\subsection{Comparison with State-of-the-art Works}
\label{subsec:comparison}

Our specification shares similarities with Cohen et al.'s Leader-Aware SMR framework (Carousel) \cite{cohen2022aware} and Liu et al.'s ABSE framework \cite{liu2025abse}: All aim to provide a generic foundation for optimizing leader election mechanisms. We and Carousel consider partial synchrony, while ABSE additionally considers asynchrony. Furthermore, all focus on designing new mechanisms through consensus-reputation-based approaches that rely solely on information available during consensus process - a self-contained design philosophy that promotes simplicity and avoids external dependencies. Our specification for correctness properties, however, differs significantly from both ABSE and Carousel.

ABSE categorizes leader election into three (two consistency one conflict) scenarios, with distinct requirements for each to ensure correctness. 
First, while covering both partial synchrony and asynchrony, ABSE's requirements bifurcate based on network assumptions, resulting in a specification that is neither concise nor intuitive.  
Second, ABSE ties correctness to specific protocols - a mechanism from ABSE cannot be proven correct in isolation but rather needs to be proven in conjunction with the target protocol after integration, i.e., each integration with a new protocol requires a new set of protocol-specific proofs. 
Additionally, ABSE's properties focus solely on correctness without quantifying effectiveness, hindering theoretical evaluation. 
Finally, conflict scenario correctness is established on rolling back original protocol flow (i.e., without ABSE), further undermining its protocol-independence characteristic. Overall, ABSE does not achieve true generality.

Carousel focuses exclusively on chain-based BFTs, narrowing its applicability compared to our specification. More critically, Carousel's specification shifts focus from leader election guarantees to block commit properties. For example, its Chain Quality property states: "For any block $B$ committed by a correct replica, the proportion of Byzantine blocks on $B$'s implied chain is bounded", 
which provides weaker guarantees than $\gamma$-Guarantee because Byzantine replicas can do more than just propose Byzantine blocks after being elected. Most simply, it fails to capture scenario where Byzantine leaders do not produce bad blocks but block progress by dropping all messages sent by correct replicas. Such a scenario can cause considerable performance degradation without violating Chain Quality. Furthermore, the focus on block commitment also weakens Carousel's protocol-independent aspects. 
Besides Chain-quality, Carousel also contains properties explicitly tied to specific execution scenarios (e.g., "crash-only executions"). Although Tsimos et al. \cite{tsimos2024hammerhead} extend Carousel's concepts to DAG-based BFTs, their formalization is explicitly restricted to DAG-based BFTs, creating framework divergence rather than unification, while inheriting the same specification limitations.

\section{Sliding Window Leader Election (SWLE)}
\label{subsec:mechanism_overview}
Following the abstraction in Section \ref{sec:abstraction}, we design the SWLE mechanism. 
The mechanism enhances leader election by dynamically adjusting leader nominations based on replica behavior.
The core design philosophy of SWLE centers on:
\begin{itemize}[leftmargin=13pt]
	\item \textbf{Reputation-based candidacy}: Leadership eligibility is determined by accumulated reputation scores reflecting consensus participation and behavior quality;
	\item \textbf{Decoupled leader finalization}: Separating leader election from current consensus rounds by finalizing future leaders (sufficiently in advance to maintain liveness after GST) through authenticated consensus decisions;
	\item \textbf{Sliding window management}: Maintaining a fixed-size window of views for future leaders while providing fairness through systematic election triggering across all replicas;
	\item \textbf{Byzantine-cost amplification}: Forcing Byzantine replicas to maintain good behavior to acquire leadership.
\end{itemize}

SWLE is designed as a generic enhancement operating atop any partially synchronous leader-based BFT that satisfies the following two fundamental Prerequisites (align with core correctness properties of most partially synchronous BFTs):

\begin{description}
	\item[Prerequisite 1 (Safety):] For any consensus round $r$, if two distinct proposals $prop_1$ and $prop_2$ both originate from $r$, they cannot be both finalized, each by a correct replica.
	
	\item[Prerequisite 2 (Liveness)]\footnote{Note that for chain-based protocols, where a proposal may span multiple consecutive rounds (i.e., traverse multiple leaders), it may require leaders are correct for $cr$ (the number of rounds needed to finalize a proposal) consecutive rounds. However, for simplicity, we omit detailed discussion of chain-based protocols in this work to maintain clarity of our fundamental approach.}: After GST, there exists a bounded interval $T_b$ such that if all correct replicas remain in consensus round $r$ during $T_b$, and the leader for $r$ is correct, then a decision is reached (i.e., a proposal is finalized) among all correct replicas within $r$.
\end{description}

It is always possible to set a timer (waiting interval) such that all correct replicas will eventually have at least $T_b$ overlap in common. Therefore, Prerequisite 2 can be equivalently stated as: After GST, if the leader for $r$ is correct, then a decision can be reached among all correct replicas within $r$.

The liveness guarantee admits an equivalent formulation: After GST, correct leadership in $r$ implies bounded-time finalization. More generally, it can be extended to state that for any sequence of $T$ consecutive rounds with at least $T_c$ rounds with correct leaders ($T \geq T_c> 0$), the decision (among all correct replicas) latency $T_f$ (measured in consensus rounds) is bounded (also, as noted, for pipelined protocols, $T_c>0$ alone may be insufficient. We also omit detailed discussion here).

SWLE leverages these protocol Prerequisites to provide enhanced guarantees. Any protocol satisfying Prerequisites 1-2 can integrate SWLE while preserving correctness. Furthermore, we prove that under SWLE, after GST, there exists a view $v_c$ beyond which the decision latency satisfies $T_f < n$.

\subsection{Mechanism Design}

SWLE maintains leader assignments across a sliding window of views. We define $T_z = \lceil T_f/n \rceil \cdot n$ as the sliding window size. Each view corresponds to exactly one consensus round. A view progresses to the next when either: (1) the current view's leader successfully completes the consensus round, or (2) the maximum waiting interval for a consensus round expires (timeout occurs).

Each replica $j$ maintains a \textsc{LeaderList} data structure of length $T_z + 2n$ (covering views $v_{\text{curr}}$ to $v_{\text{curr}} + T_z + 2n - 1$, where $v_{\text{curr}}$ is the current view) which stores leader information for consecutive views. Initially, \textsc{LeaderList} covers views 1 to $T_z + 2n$ and is partitioned into two segments:
\begin{itemize}
	\item For $v$ in $1$ to $T_z + n$: The \textbf{Elected Leader} ($\textsc{{leader}}_j(v)$) and the \textbf{Initial Leader} (i\textsc{leader}$_j(v)$) of $v$ are both initialized as $v\mod n$.
	\item For $v$ in $T_z + n + 1$ to $T_z + 2n$: $\textsc{{leader}}_j(v)=null$, i\textsc{leader}$_j(v)=v\mod n$
\end{itemize}

When entering $v$, $\textsc{{leader}}_j(v)$ and i\textsc{leader}$_j(v)$ are no longer modified. $j$ will first determine either $\textsc{leader}_j(v)$ if $\textsc{leader}_j(v) \neq null$, or i$\textsc{leader}_j(v)$ otherwise as $v$'s leader. If a replica $k$ gathers $2f+1$ valid votes for leadership claims (i.e., proposes a proposal containing a valid leader certificate, see Algorithm \ref{alg:cs}) in $v$, $j$ will then determine $k$. On completing $v$, $j$ remove $v$ from \textsc{LeaderList}, and then extends it by initializing $v_{\text{new}} = v + T_z + 2n$ such that $\textsc{{leader}}_j(v_{\text{new}})=null$ and i\textsc{leader}$_j(v_{\text{new}})=v_{\text{new}}\mod n$ (obviously, for any view, all correct replicas will have the same initial leader).

SWLE tracks replica behavior during consensus. Each replica $j$ maintains a local reputation matrix $S_j \in \mathbb{R}^{n \times 1}$ where $S_j[i]\geq 0$ represents $j$'s assessment of replica $i$'s reputation score.
Table \ref{tab:scoring_rules} summarizes the scoring mechanism employed by SWLE. For $j$, $\text{Score}_j$ is updated according to these rules:

\begin{table}[h]
	\caption{SWLE Reputation Scoring Rules}
	\label{tab:scoring_rules}
	\vspace{-1mm}
	\begin{tabular}{|c|c|c|c|}
		\hline
		\textbf{Rule} & \textbf{Trigger Event} & \textbf{Affected Replica} & \textbf{Score Change{$^{3,4}$}} \\
		\hline
		\textbf{R1} & Enter new view $v$ & Leader{$^5$} of $v$ & $^6$ $\alpha_0=-1$  \\
		\hline
		\textbf{R2} & View $v$ timeout & Leader{$^5$} of $v$ & $\alpha_1=-n$  \\
		\hline
		\textbf{R3} & Proposal finalized & Proposal's leader & $\alpha_2=+1$  \\
		\hline
		\textbf{R4} & \begin{tabular}{@{}c@{}}Complete view $v$\\ successfully as\\ $v$'s leader.\end{tabular} & \begin{tabular}{@{}c@{}}Consensus-promoting \\ replicas{$^7$} of $v$ \end{tabular} & $\alpha_3=+1/n$ \\
		\hline
		\textbf{R5} & \begin{tabular}{@{}c@{}}{$^4$} $v\bmod\Theta=0$ \\ $\Theta $$= $$\text{max}(300,$$ 10n)$\end{tabular} & All replicas & $\alpha_4=+1$ \\
		\hline
	\end{tabular}
 	\\ 
 	\\{${}^3$When rule Ri is triggered at replica $j$, for replica $k$ in Ri.Affected\_Replicas: $S_j[k] = \max(S_j[k] + \alpha_i, 0)$}.  
 	\\{${}^4$Values can be adjusted (but note that this may affect {\scriptsize SWLE's} correctness and effectiveness). Initially, $S_j[k] =|\alpha_4|$ for all $j,k\in\mathcal{R}$.} 	
 	\\{${}^5$Determination priority: Replica that gathers $2f+1$ valid votes for leadership claims $>$ Elected Leader $>$ Initial Leader.}
 	\\ {${}^6$$\forall k$$\in$$\mathcal{R}$, $S_j[k]$$\ge$$ |\alpha_0|$ is required to maintain $k$'s leadership eligibility at $j$.}
 	\\${}^7$Any replica that contributes one of the first $2f+1$ valid votes.
\end{table}

Where R1 is the score deduction rule for view entry (leader accountability), R2 is the score deduction rule for a failed leadership (timeout penalty), R3 is the score reward rule for proposal finalization (success reward), and R4 is score reward rule for timely voting (participation reward). These rules ensure  Byzantine leaders incur score deductions, which require subsequent good behavior to offset. R5 is the periodic normalization rule that, along with non-negative score constraint, prevents permanent marginalization and maintains system adaptability. This is particularly important in scenarios where, before GST, Byzantine replicas might persistently target certain correct replicas, causing their scores to decrease indefinitely at others, making them hard to recover leadership eligibility even after GST (i.e., the time is unbounded).

In SWLE, replicas collaboratively elect leaders for future views through authenticated consensus decisions. It employs a sophisticated candidate selection algorithm that ensure fairness through systematic election triggering across all replicas. Replica $j$ include additional fields $l_d$ (determined leader for $v$) and $\text{Cand}_j(v_{\text{target}})$ (leadership candidates for $v_{\text{target}}$) in all voting messages for view $v$ (including votes during view-change to $v+1$). For the votes to be valid, besides standard consensus checks, $l_d$ must be non-empty and contain a valid ID, $\text{Cand}(v_{\text{target}})$ must be non-empty, non-duplicated and contain valid IDs. The target view $v_{\text{target}}$ is computed as follows:

\vspace{-2mm}
\begin{align}
	\nonumber
	v_{\text{target}} = \begin{cases} 
		v + n + T_z + \left(\left\lfloor v/n \right\rfloor \bmod n\right) & \text{if Case (1)} \\
		v + n + T_z & \text{if Case (2)} \\
		v + T_z + \left(\left\lfloor (v-1)/n \right\rfloor \bmod n\right) & \text{otherwise}
	\end{cases}
\end{align}

\begin{itemize}[topsep=0pt,itemsep=0pt,parsep=0pt]
		\vspace{-0.5mm}
	\item Case (1): $v + \left(\left\lfloor v/n \right\rfloor \bmod n\right) \leq \left\lceil v/n \right\rceil \cdot n$
	\item Case (2): $v/n \in \mathbb{Z} \land (v/n) \bmod n = 1$
\vspace{-0.5mm}
\end{itemize}

\begin{lemma}\label{lem:original}
	Let $n$ be a positive integer and $x$ a natural number, when $v$ traverses the set $\{nx + 1, \dots, n(x+1)\}$, the values $v_{\text{target}}$ traverse the set $\{n(x+1) + 1, n(x+1) + 2, \dots, n(x+2)\}$ exactly once.
\end{lemma}

\begin{proof}
	Let $r = v - nx$ and $a = x \mod n$. Since $v \in \llbracket nx + 1, n(x+1) \rrbracket$, we have $r \in \{1, 2, \dots, n\}$ and $a \in \{0, 1, \dots, n-1\}$. The expression for $v_{\text{target}}$ simplifies to:
	\[
	v_{\text{target}} = 
	\begin{cases} 
		nx + n + (r + a) & \text{if } r + a \leq n \\
		nx + r + a & \text{if } r + a > n
	\end{cases}
	\]
	The target interval is $\{t_k \mid t_k = nx + n + k,\  k = 1, 2, \dots, n\}$. We prove the mapping $v \mapsto v_{\text{target}}$ is a bijection onto this set.
	
	\vspace{1mm}
	\noindent \underline{{\bfseries Step 1: Validity of $v_{\text{target}}$ in the target interval.}}
	\begin{itemize}
		\item If $r + a \leq n$, then $v_{\text{target}} = nx + n + (r + a)$. Since $r \geq 1$ and $a \geq 0$, we have $r + a \geq 1$, so $v_{\text{target}} \geq nx + n + 1$. Since $r + a \leq n$, we have $v_{\text{target}} \leq nx + n + n = nx + 2n$.
		\item If $r + a > n$, then $v_{\text{target}} = nx + r + a$. Since $r + a \geq n + 1$, we have $v_{\text{target}} \geq nx + n + 1$. Since $r \leq n$ and $a \leq n-1$, we have $r + a \leq 2n - 1$, so $v_{\text{target}} \leq nx + 2n - 1 \leq nx + 2n$.
	\end{itemize}
	Thus, $v_{\text{target}} \in \{nx + n + 1, \dots, nx + 2n\}$ for all $v$.
	
	\vspace{2mm}
	\noindent\underline{{\bfseries Step 2: Surjectivity (coverage).}}
	For each $k \in \{1, 2, \dots, n\}$, we construct $r \in \{1, 2, \dots, n\}$ such that $v_{\text{target}} = t_k \triangleq nx + n + k$:
	\begin{itemize}
		\item If $k - a \geq 1$, set $r = k - a$. Then $r + a = k \leq n$, so $v_{\text{target}} = nx + n + (r + a) = nx + n + k = t_k$. Also, $r = k - a \geq 1$ and $r = k - a \leq k \leq n$.
		\item If $k - a \leq 0$ (i.e., $k \leq a$), set $r = n + k - a$. Since $k \leq a \leq n-1$, we have $k \leq n-1$. Now:
		\[
		r + a = (n + k - a) + a = n + k > n \quad (\text{since } k \geq 1),
		\]
		so $v_{\text{target}} = nx + r + a = nx + n + k = t_k$. Further:
		\[
		r = n + k - a \geq n + k - (n-1) = k + 1 \geq 1 \quad (\text{since } k \geq 1),
		\]
		and since $k \leq a$,
		$r = n + k - a \leq n + a - a = n.$
	\end{itemize}
	Thus, every $t_k$ is covered.
	
	\vspace{2mm}
	\noindent\underline{{\bfseries Step 3: Injectivity (no overlaps).}}
	Suppose two distinct $r_1, r_2 \in \{1, \dots, n\}$ map to the same $v_{\text{target}}$. Consider cases:
	\begin{itemize}
		\item If both satisfy $r_i + a \leq n$, then:
		\[
		nx + n + (r_1 + a) = nx + n + (r_2 + a) \implies r_1 = r_2.
		\]
		\item If both satisfy $r_i + a > n$, then:
		\[
		nx + r_1 + a = nx + r_2 + a \implies r_1 = r_2.
		\]
		\item If one satisfies $r_1 + a \leq n$, the other $r_2 + a > n$, then:
		\[
		nx + n + (r_1 + a) = nx + r_2 + a \implies n + r_1 = r_2.
		\]
		Since $r_1 \geq 1$, we have $r_2 \geq n + 1 > n$, contradicting $r_2 \leq n$.
	\end{itemize}
	Thus, distinct $v$ map to distinct $v_{\text{target}}$.
	
	Since the mapping is injective and surjective onto the target set of size $n$, it is a bijection.
\end{proof}

\begin{proposition}\label{prop:1}
	Let $T_z$ be a positive integer multiple of $n$. For any natural number $x$ and positive integer $n$, when $v$ traverses $\{nx + 1, \dots, n(x+1)\}$, the values $v_{\text{target}}$ traverse $\{n(x+1) + 1 + T_z, \dots, n(x+2) + T_z\}$ exactly once.
\end{proposition}

\begin{proof}\let\qed\relax
	Observe the shifted target $v_{\text{target}}$ can be expressed as:
	\[
	v_{\text{target}} = v_{\text{target}}^0 + T_z
	\]
	where $v_{\text{target}}^0$ is the original mapping defined in Lemma \ref{lem:original}. This follows from direct comparison of the cases:
	\begin{itemize}
		\item In \textit{"if Case (1)"}: 
		\[
		v_{\text{target}} = \underbrace{v + n + \left(\left\lfloor v/n \right\rfloor \mod n\right)}_{\text{original } v_{\text{target}}^0} + T_z
		\]
		\item In \textit{"else if Case (2)"}: 
		\[
		v_{\text{target}} = \underbrace{v + n}_{\text{original } v_{\text{target}}^0} + T_z
		\]
		\item In \textit{"otherwise"}: 
		\[
		v_{\text{target}} = \underbrace{v + \left(\left\lfloor (v-1)/n \right\rfloor \mod n\right)}_{\text{original } v_{\text{target}}^0} + T_z
		\]
	\end{itemize}
	
	By Lemma \ref{lem:original}, as $v$ traverses $\{nx + 1, \dots, n(x+1)\}$, the values $v_{\text{target}}^0$ traverse $\mathscr{T} \triangleq \{n(x+1) + 1, \dots, n(x+2)\}$ bijectively. Since $T_z$ is a constant shift:
	\[
	v_{\text{target}} = v_{\text{target}}^0 + T_z
	\]
	traverses the shifted set $\{t + T_z \mid t \in \mathscr{T}\} = \{n(x+1) + 1 + T_z, \dots, n(x+2) + T_z\}$ bijectively.
	
	The shift preserves bijectivity because:
	\begin{itemize}
		\item \textbf{Surjectivity:} For any $s \in \{n(x+1) + 1 + T_z, \dots, n(x+2) + T_z\}$, let $t = s - T_z$. Then $t \in \mathscr{T}$, so there exists $v$ with $v_{\text{target}}^0 = t$, giving $v_{\text{target}} = t + T_z = s$.
		\item \textbf{Injectivity:} If $v_{\text{target}}(v_1) = v_{\text{target}}(v_2)$, then $v_{\text{target}}^0(v_1) + T_z = v_{\text{target}}^0(v_2) + T_z$ implies $v_{\text{target}}^0(v_1) = v_{\text{target}}^0(v_2)$, so $v_1 = v_2$ by Lemma \ref{lem:original}.
	\end{itemize}
\end{proof}
 
Proposition \ref{prop:1} guarantees that the election process of each view (except the initial $T_z+n$ views) is initiated exactly once, preventing both election skips and duplicates, while ensuring fair election triggering across all replicas (i.e., every replica has equal opportunity to initiate elections for views where any given replica serves as the initial leader) thus further reducing the leverage that Byzantine replicas can exert over the election process (e.g., the chance of Byzantine replicas controlling election initiation of Byzantine replicas is reduced) and increasing the cost of manipulation for Byzantine replicas.

A replica $j$ generates $\text{Cand}(v_{\text{target}})$ in $v$ as in Algorithm \ref{alg:cag}:

\begin{algorithm}[H]
	\caption{Candidate Array Generation}
		\label{alg:cag}
	\small
	\begin{algorithmic}
		\State Initialization: $\text{Cand}_j(v_{\text{target}}) \leftarrow \emptyset$
	\end{algorithmic}
\vspace{-1mm}
	\begin{algorithmic}[1]
		\For{$i = 0$ to $n-1$}
		$view \leftarrow v_{\text{target}} + i$
		\State $candidate \leftarrow \text{i}\textsc{leader}_j(view)$
		\If{$S_j[candidate] \geq |\alpha_0|$}
		\State $\text{Cand}_j(v_{\text{target}}) \leftarrow \text{Cand}_j(v_{\text{target}}) \cup \{\text{candidate}\}$
		\EndIf
		\EndFor
		\If{$\text{Cand}_j(v_{\text{target}}) = \emptyset$}
		\State  \Comment{apply periodic normalization (individually) and recompute}
		\State \textbf{for} $k \in \mathcal{R}$: \textbf{do} $S_j[k]$  $\leftarrow S_j[k] + \alpha_4$
		\State \textbf{goto} line 1
		\EndIf
		\State \textbf{return} $\text{Cand}_j(v_{\text{target}})$
	\end{algorithmic}
	\vspace{-1mm}
\end{algorithm} 

The leader of view $v+1$ collects $2f+1$ valid votes containing $\text{Cand}(v_{\text{target}})$ values and the same $l_d$ matching the leader's ID. In chain-based protocols, votes for $v$ are inherently sent to the leader of $v+1$. In non-chained protocols, 
these may be votes from standard phases for $v$ or view-change process to $v+1$. The leader determination process is then performed as shown in Algorithm \ref{alg:cs}.

\begin{algorithm}[H]
	\caption{Candidate Selection and Proof Packaging}
	\label{alg:cs}
	\small
	\begin{algorithmic}
		\State Required: $\mathcal{V}$ \Comment{set of $2f$$+$$1$ valid votes containing $\text{Cand}(v_{\text{target}})$} 
		\State Initialization: $\text{count}[i] \leftarrow 0$ for all $i \in \mathcal{R}$
	\end{algorithmic}
	\begin{algorithmic}[1]
		\State $\mathcal{C}\leftarrow \{vote.\text{Cand}_k(v_{\text{target}})| vote\in\mathcal{V}\} $ \Comment{collect candidate arrays}
		\For{$\text{Cand}_k(v_{\text{target}}) \in \mathcal{C}$}
		\For{each candidate $c \in \text{Cand}_k(v_{\text{target}})$}
		\State $\text{count}[c] \leftarrow \text{count}[c] + 1$
		\EndFor
		\EndFor
		\State $\mathcal{W} \leftarrow \{c | \text{count}[c] \geq f+1\}$ \Comment{find c within $\geq f+1$ votes.}
		\If{$\mathcal{W} \neq \emptyset$}
		\State Select $l \leftarrow \arg\min_{c \in \mathcal{W}} \{v'| \text{i\textsc{leader}}(v') = c \land v' \geq v_{\text{target}}\}$
		\Else \ $l \leftarrow \bot$
		\EndIf
		\State \textbf{return} $\text{LCert}_{v_{\text{target}}} \leftarrow \langle l, \mathcal{V} \rangle$ \Comment{create a leader certificate}
	\end{algorithmic}
	\vspace{-1mm}
\end{algorithm}

The leader certificate $\text{LCert}_{v_{\text{target}}}$ contains both the selected leader and cryptographic proof (the original $2f+1$ votes). A replica can verify it by performing the same computation as the leader when creating it. $\text{LCert}_{v_{\text{target}}}$ is embedded in new proposals for $v+1$ (as an additional field, though in chain-based protocols, this can be integrated into existing quorum certificates \cite{yin2019hotstuff} that already contain voting proofs created by votes for the previous view), and its validity should be considered during standard proposal validation.

When a proposal containing $\text{LCert}_{v_{\text{target}}}$ is finalized, each replica $j$ updates its \textsc{LeaderList} as in Algorithm \ref{alg:lup}.

\begin{algorithm}[H]
	\caption{\textsc{LeaderList} Update}
	\label{alg:lup}
	\small
	\begin{algorithmic}
		\State Required: $\text{LCert}_{v_{\text{target}}}$
	\end{algorithmic}
	\begin{algorithmic}[1]
		\State Extract $\langle l, \mathcal{V} \rangle$ from $\text{LCert}_{v_{\text{target}}}$
		\If{$l = \bot$}
		$\textsc{leader}_j(v_{\text{target}}) \leftarrow \text{i\textsc{leader}}_j(v_{\text{target}})$ \Comment{fallback}
		\Else \  $\textsc{leader}_j(v_{\text{target}}) \leftarrow l$ 
		\EndIf
		\Statex \Comment{$v^*$ is the highest view lower than $v_{\text{target}}$ whose leader is elected}
		\State $v^* \leftarrow \arg\max_v\{v' < v_{\text{target}} | \textsc{leader}_j(v') \neq \text{null}\}$
		\For{$v' = v^* + 1$ to $v_{\text{target}} - 1$} \Comment{fill intermediate views}
		\State $\textsc{leader}_j(v') \leftarrow \text{i\textsc{leader}}_j(v')$
		\EndFor
	\end{algorithmic}
	\vspace{-1mm}
\end{algorithm}

Lines 4-6 fill gaps using default initial leaders to ensure continuity in \textsc{LeaderList}, trying to maintain progress particularly in pre-GST or Byzantine leader scenarios. The integration of election results into consensus finalizations ensures leader updates are atomic with consensus decisions, eliminating race conditions between leader election and view progression. Meanwhile, consensus safety prevents inconsistent leader assignments within a view across correct replicas.

\subsection{Mechanism Analysis}
	
	We now provide formal correctness analysis for SWLE. 
	
	\begin{theorem}[Leader Uniqueness]
		\label{thm:leader-uniqueness}
		For any view $v$, there do not exist two distinct replicas $j$ and $k$  both can gather $2f+1$ distinct valid votes for conflicting leadership claims in $v$.
	\end{theorem}
	
	\begin{proof}
		We prove this by contradiction. Suppose there exist two distinct replicas $p$ and $q$ that both gather $2f+1$ distinct votes for conflicting leadership claims in view $v$.
		\begin{itemize}
			\item By Prerequisite 1, no two conflicting proposals originate from the same consensus round can be both finalized, each by a correct replica. Therefore, there cannot be two correct replicas $j$ and $k$ such that $\textsc{leader}_j(v)=p, \textsc{leader}_k(v)=q$. 
			\item Votes from any correct replica contain the leader ($l_d$ field) that it determines. A Byzantine replica cannot use votes that determine other replicas for its own leadership claim.
			\item Any correct replica $j$ only determine either $\textsc{leader}_j(v)$ (if $\textsc{leader}_j(v)\ne\bot$) or i$\textsc{leader}_j(v)$ (otherwise) as $v$'s leader before a replica gathers $2f+1$ valid votes for leadership claims in $v$. Without loss of generality, we assume that $p$ first gathers $2f+1$ valid votes and claims (proposes a proposal containing a valid leader certificate). At least $f+1$ of the $2f+1$ votes are from correct replicas, whose votes in $v$ only determine $p$. Thus, $q$ can only gather votes (which determine $q$ and are valid) from at most $2f$ replicas (of which $f$ are correct), leading to a contradiction.
		\end{itemize}
		Therefore, Leader Uniqueness holds.
	\end{proof}
	
	\begin{lemma}
		\label{lem:1}
		For any correct replica \(j\), if a proposal in view \(v\) finalizes the leader election for target view \(v_{\text{target}}\), then for all views \(1 \leq \textit{view} \leq v_{\text{target}}\), \(\textsc{leader}_j(\textit{view}) \neq \text{null}\).
	\end{lemma}
	
	\begin{proof}
		By the Algorithm \ref{alg:lup}, whenever a finalizing proposal containing \(\text{LCert}_{v_{\text{target}}}\) is processed:  
		
		\noindent 1. If the elected leader \(l \neq \bot\), \(\textsc{leader}_j(v_{\text{target}})\) is set to \(l\). \\ 
		2. If \(l = \bot\), \(\textsc{leader}_j(v_{\text{target}})\) defaults to \(\text{i\textsc{leader}}_j(v_{\text{target}})\).  
		3. The algorithm then fills all intermediate views \(v^* + 1\) to \(v_{\text{target}} - 1\) (where \(v^*\) is the highest view \(< v_{\text{target}}\) with non-null leader) by setting \(\textsc{leader}_j(v') = \text{i\textsc{leader}}_j(v')\) for each \(v'\).  
		
		By Proposition \ref{prop:1}, leader elections are uniquely triggered for every view beyond the initial \(T_z + n\) views. Combined with the initialization of \(\textsc{leaderlist}\) (covering views \(1\) to \(T_z + n\)) and the gap-filling step in Algorithm \ref{alg:lup}, all views \(\leq v_{\text{target}}\) are guaranteed to have non-null leader assignments.  
	\end{proof}
	
	\begin{lemma}
		\label{lem:2}
		After GST, there exists a view \(v_c\) such that for any correct replica \(j\) and any \(v_j > v_c\) where \(i\textsc{leader}_j(v_j) = j\), \(j\) gathers \(2f+1\) votes for its leadership claim in \(v_j\).  
	\end{lemma}
	
	\begin{proof}
		Let \(v_{\text{GST}}\) be the first view after GST, and \(\Theta = \max(300, 10n)\). Define \(v_1\) as the lowest view \(> v_{\text{GST}}\) such that \(v_1 \mod \Theta = 0\). By Rule R5, upon entering \(v_1\), all correct replicas add \(\alpha_4 = +1\) to all replicas' scores. Since scores are non-negative in SWLE, after \(v_1\), every correct replica \(j\) has \(S_k[j] \geq 1\) at all correct replicas \(k\).  
		
		For the first view \(v_j\) with \(\text{i\textsc{leader}}_k(v_j) = j\) such that \(v'_j>v_1\) where \(v'_j\) is the view in which \(v_j\)'s election occurs (let $v_c=\arg\min_v\{v_j | j\in \mathcal{R}\}-1$). 
		Since \(S_k[j] \geq |\alpha_0|=1\) at all correct \(k\), Algorithm \ref{alg:cag} includes \(j\) in \(\text{Cand}_k(v_j)\).  The leader of \(v'_j+1\) can (but may not, e.g. the leader is Byzantine) collect \(\geq 2f+1\) votes. At least \(f+1\) votes must be from correct replicas (since at most $f$ replicas are Byzantine), all containing \(j\) in their candidate sets.  By Algorithm \ref{alg:cs}, \(j\) receives \(\geq f+1\) votes and is selected as leader (since \(v_j\) is the minimal view from \(v_j\) with initial leader \(j\)). The leader certificate \(\text{LCert}_{v_j}\) containing both $j$ and the original $2f+1$ votes should be embedded in the proposal for \(v'_j+1\) to make the proposal valid. For any correct $k$, three cases arise: 
		\begin{itemize}
			\item If \(v'_j+1\)'s proposal finalizes before \(v_j\), then before entering \(v_j\), $k$ sets \(\textsc{leader}_k(v_j) = j\) (Algorithm \ref{alg:lup}) and $k$ determines \(j\) in \(v_j\). 
			\item If a later proposal \(v' > v'_j+1\) finalizes before \(v_j\), according to Lemma \ref{lem:1}, $k$ also has \(\textsc{leader}_k(v_j) = j\) before entering \(v_j\), and $k$ determines \(j\) in \(v_j\). 
			\item Otherwise, when $k$ enters \(v_j\), $\textsc{leader}_k(v_j) = null$. According to SWLE $k$ defaults to \(\text{i\textsc{leader}}_k(v_j) = j\) and determines \(j\) in \(v_j\). 
		\end{itemize}
		In all cases, $j$ can collect $2f+1$ votes for its leadership claim in \(v_j\), after which all correct replicas determine $j$ in \(v_j\). By Theorem \ref{thm:leader-uniqueness}, there does not exist another replica that can gather $2f+1$ votes for its leadership claim. When $j$ serves as the unique leader, by Prerequisite 2, the proposal of view \(v_j\) can be finalized by all correct replicas before timeout. Therefore, $j$'s net score change at correct replicas is 0 (entering view: $\alpha_0 = -1$, proposal finalized: $\alpha_2 = +1$). By induction, this property extends to all subsequent views where $j$ serves as initial leader.
	\end{proof}
	
	\begin{lemma}
		\label{lem:3}
		After GST, there exists a view \(v_c\) such that after \(v_c\), the duration \(T_f\) (in views) for a decision to be reached among all correct replicas is bounded with \(T_f < n\).  
	\end{lemma}
	
	\begin{proof}
		By Lemma \ref{lem:2}, there exists a view $v_c$ such that after $v_c$, any window of \(n\) consecutive views includes at least $2f+1$ correct leaders (since in any $n$ consecutive views, each correct replica serves as initial leader for exactly one view, and there are $\geq 2f+1$ correct replicas). Each correct leader can finalize proposals within bounded time after GST (Prerequisite 2). Even if Byzantine leaders delay progress, within any sequence of $f+1$ consecutive views, there must be at least one correct leader, ensuring progress. Therefore, $T_f$ is bounded and $T_f < n$.
	\end{proof}

	\begin{lemma}
		\label{lem:4}
		For any view \(v\), \(v_{\text{target}} - (v+1) \geq T_z\) (or $v_{\text{target}} - v \geq T_z + 1$), where \(T_z = \lceil T_f / n \rceil \cdot n\). 
	\end{lemma}
	
	\begin{proof}
		We analyze the target view computation based on the three cases in the SWLE design:
		\begin{itemize}
			\item When $v + (\lfloor v/n \rfloor \bmod n) \leq \lceil v/n \rceil \cdot n$:
			$$v_{\text{target}} - v = T_z + n + (\lfloor v/n \rfloor \bmod n) \geq T_z + n \geq T_z + 1$$
			since $n \geq 1$.
			\item Else when \(v/n \in \mathbb{Z}\) $\land$ $(v/n) \bmod n = 1$:
			$$v_{\text{target}} - v = T_z + n \geq T_z + 1$$
			\item Otherwise:$$v_{\text{target}} - v = T_z + (\lfloor (v-1)/n \rfloor \bmod n)$$
			
			Since $v + (\lfloor v/n \rfloor \bmod n) > \lceil v/n \rceil \cdot n > v$, we have $(\lfloor v/n \rfloor \bmod n) \geq 1$. Given that $(v/n) \bmod n \neq 1$, we have $(v/n) > \lfloor v/n \rfloor \geq 1$, and since both $v$ and $n$ are positive integers, $(v-1)/n \geq \lfloor v/n \rfloor \geq 1$, i.e.,  $(\lfloor (v-1)/n \rfloor \bmod n)\ge 1$
			$$v_{\text{target}} - v \geq T_z + 1$$. Therefore:
		\end{itemize}
		In all cases, $v_{\text{target}} - v \geq T_z + 1$.
	\end{proof}
	
	\begin{theorem}[Timely Finalization]
		\label{thm:timely-finalization}
		Under $f$ faults out of $3f+1$ replicas, after GST, there exists a view $v_c$ such that for any view $v > v_c$, every correct replica $j$ finalizes $\textsc{leader}_j(v)$ and i$\textsc{leader}_j(v)$ before or upon entering $v$.
	\end{theorem}
	
	\begin{proof}
		In SWLE, it is evident that, once $\textsc{leader}_j(v)$ is finalized, i$\textsc{leader}_j(v)$ must already have been finalized. Therefore, we only need to consider $\textsc{leader}_j(v)$.
		
		By Lemma \ref{lem:3}, after GST, there exists $v_c$ such that after $v_c$, the decision latency $T_f$ (in views) is bounded with $T_f < n$. Assume for contradiction that after $v_c$, a correct replica $j$ enters a view $v' > v_c$ without finalizing $\textsc{leader}_j(v')$. By Proposition \ref{prop:1},  $v'$'s leader election must be initiated in a prior view $v_{elec}<v'$. By Lemma \ref{lem:1}, this implies no proposal is finalized for views between $v_{elec}+1$ to $v'$.
		
		By Lemma \ref{lem:4}:
		The gap between $v_{elec}+1$ and $v'$ satisfies $v' - (v_{elec}+1) \geq T_z $, where $T_z = \lceil T_f / n \rceil \cdot n$.
		However, by Lemma \ref{lem:3}, within any $T_f$ consecutive views after $v_c$, at least one proposal is finalized by all correct replicas. Since $T_f \leq T_z$, the maximum interval between finalized proposals is $T_f \leq T_z$. This contradicts $v' - (v_{elec}+1) \geq T_z $.
	\end{proof}
	
	\begin{theorem}[$\gamma$-Guarantee]
		\label{thm:gamma-guarantee}
		After GST, there exists a view $v_c$ such that SWLE satisfies the $\gamma$-Guarantee with parameters $(v_c, T = n, \gamma = \frac{2f+1}{n}, \sup = n -  (n/\Theta\cdot f)(1+T_z/n)$).
	\end{theorem}
	
	\begin{proof}
		By Lemma \ref{lem:2}, after \(v_c\), for any view \(v_j > v_c\) with \(\text{i\textsc{leader}}(v_j) = j\) (replica \(j\) is correct), \(j\) gathers \(2f+1\) votes for its leadership cliam in \(v_j\). All correct replicas agree on (detemine) \(j\) in $v_j $and \(j\) maintain zero net score change (score decreases by $\alpha_0 = -1$ upon view entry and increases by $\alpha_2 = +1$ upon proposal finalization).
		
		After $v_c$, we analyze the behavior of Byzantine replicas when they serve as leaders:
		\begin{itemize}[leftmargin=13pt]
			\item Correct behavior (similar to correct replicas): Net score change is $0$, maintaining reputation score $\ge |\alpha_0|=1$ and leadership eligibility at all correct replicas .
			\item Misbehavior (e.g., timeout or invalid proposals): Score decreases by $\alpha_0+ \alpha_1 = -(n+1)$ (Rule R1 and R2), potentially reaching $<|\alpha_0|=1$ (losing leadership eligibility) at correct replicas.
		\end{itemize}

		In any sequence of \(T = n\)-view window (corresponding to $n$ consecutive views), the initial leader assignment covers all $n$ replicas exactly once due to i$\textsc{leader}(v) = v \bmod n$. Therefore, $T = n$. By Lemma \ref{lem:2} and Theorem \ref{thm:timely-finalization}, for any view $v'>v_c$ where the initial leader is a correct replica $j$, all correct replicas agree on (determine) the same leader $j$ before or upon entering $v'$ (i.e., $\textsc{leader}_k(v') = j$ for all correct replicas $k$). Since there are $2f+1$ correct replicas, the lower bound for the expected number of views with unified correct leadership is $\gamma T = \frac{2f+1}{n} \cdot n = 2f+1$ (i.e, \(\gamma = \frac{2f+1}{n}\)).

		For the upper bound $\sup$, we analyze Byzantine replica elections. For any view $v''>v_c$ where the initial leader is a Byzantine replica $m$, it may be elected or defeated (the elected leader of $v''$ may be replaced by a correct replica when defeated). However, at very least, 
		every $\Theta$ views, Rule R5 increases all replicas' scores by $\alpha_4 = 1$. 
		For any Byzantine replica $m$, after this score increases, it is certain to win the election in $v'_m$, when, for the first time, the target view is $v_m$, provided that $m$ is the initial leader for $v_m$, because:
		
		\begin{enumerate}
			\item No rule decreases its score before it serves as leader.
			\item \(S_k[m] \geq 1\) at all correct \(k\), so \(m \in \text{Cand}_k(v_m)\).
			\item $v_m$ is the minimum view number from $v_m$.
		\end{enumerate}
		Similar to the analysis of correct replicas after $v_c$, all correct $k$ agree on $m$ as $\textsc{leader}_k(v_m)$ before or upon entering $v_m$. 
		Similarly, $m$ is certain to win elections initiated in views after $v'_m$ (where $m$ serves as the initial leader of the target views) for at least $\min\lfloor (v_m-(v'_m+1))/n \rfloor = \min\lfloor (v_{target}-(v+1))/n \rfloor = T_z/n$ (Lemma \ref{lem:4}) times. Thus, each Byzantine replica leads at least $(1+T_z/n)$ times every \(\Theta\) views. 
		In \(T = n\)-view window,  for $f$ Byzantine replicas, they lead at least $(n/\Theta\cdot f)(1+T_z/n)$ views in expectation. Hence, the upper bound is:
		$\sup = n -  (n/\Theta\cdot f)(1+T_z/n)$.
	\end{proof}

\noindent \textbf{Complexity Analysis.} SWLE introduces minimal overhead.\\
(1) Each replica independently maintains its local reputation matrix $S$ and $\textsc{leaderlist}$ without inter-replica communication. The leader determination, scoring and candidate selection processes (all simple computations, though the overall overhead scales with replica count) are performed locally introducing no extra Communication Complexity. \\
(2) Each replica $j$ maintains a constant-size \(S_j \in \mathbb{R}^{n \times 1}\) and a bounded \(\textsc{leaderlist}\) of size $T_z + 2n$ - both sizes are fixed and protocol-independent (extra Space Complexity).\\
(3) Voting messages carry two extra fields: \(l_d\) (leader ID, size \(O(1)\)), \(\text{Cand}\) (candidate array, size \(O(n)\)), while proposals carry $\text{LCert}$. For protocols already include voting proofs or certificates from previous views, $\text{LCert}$ can be directly integrated with only $l$ as extra overhead (Communication Overhead).
\vspace{1mm}

\noindent \textbf{Impact of Byzantine Replicas.} \(S\) and \(\textsc{leaderlist}\) are locally maintained, Byzantine replicas cannot directly interfere with correct replicas' scoring or leader determinations. While Byzantine leaders may selectively include votes when forming $\text{LCert}$, valid certificates require at least $f+1$ votes from correct replicas, and this only affects their own $S$ regarding scoring (but they already can arbitrarily manipulate their own scores regardless of rules). Pre-GST, Byzantine replicas may target certain correct replicas through message delays, reducing their scores to $<|\alpha_0|$ at others to eliminate their leadership eligibility (but liveness cannot be originally guaranteed pre-GST, i.e., they can still achieve  similar effect without SWLE). However, SWLE's design ensures rapid recovery:  after GST, there exists $v_c$ beyond which correct replicas' worst-case election probability never falls below conventional mechanisms (a nearly optimal lower bound for reputation-based approaches since Byzantine replicas can pretend to be correct, as discussed in Section \ref{subsec:prospe}) regardless of Byzantine behavior, while achieving significantly higher upper bounds.
\vspace{1mm}

\noindent \textbf{Cost of Byzantine Behavior.} In conventional mechanisms, Byzantine replicas simply wait for their turn to become leaders and disrupt consensus progress without contributing positively. SWLE penalizes leadership misbehavior by significant score reductions ($\alpha_1$$ = $$-n$). Once a Byzantine replica loses eligibility (score$<$$|\alpha_0|$$=$$1$), its position may be replaced by correct ones. Regaining eligibility\footnote[8]{Byzantine replicas may also achieve leadership via certain uneconomical strategies without first regaining eligibility at correct replicas (after GST). 
For example, malicious replicas may attempt to regain leadership eligibility via collusion. Suppose Byzantine replica $p$ loses eligibility at correct replicas. For a target view $v_{\text{target}}$ where i$\textsc{leader}(v_{\text{target}}) = p$ ($v_{\text{target}}$ is computed from $v$) and view $v+1$ with replica $q$ as the leader, since correct replicas will not include $p$ in their candidate arrays in votes for $v$, $q$ cannot collect $2f+1$ valid votes where $f+1$ contain $p$ as a candidate. If $q$ is also Byzantine, after GST, for $p$ to become $v_{\text{target}}$'s leader recognized by correct replicas, $q$ could deliberately cause its proposal in $v+1$ to fail (e.g., by proposing an invalid proposal or timing out). This prevents finalizing a leader certificate for $v_{\text{target}}$, forcing correct replicas to fallback to i$\textsc{leader}(v_{\text{target}}) = p$. However, this strategy forces $q$ to incur a significant penalty ($\alpha_1 = -n$ score reduction via rule R2) at correct replicas to grant $p$ a single leadership view—a net loss for the adversary (sacrificing $\geq n$ reputation points but essentially just trading leadership opportunities between Byzantine replicas). Such collusion is therefore generally uneconomical and not considered in our primary analysis.
} 
at correct replicas can only through: (a) having its valid proposals finalized (R3), (b) contributing as one of the first $2f+1$ valid votes in successful consensus rounds led by correct replicas (R4), or (c) waiting for R5 be triggered. Since (a) becomes unavailable after losing eligibility, Byzantine replicas must rely on (c) which has a very long cycle ($\ge$$10n$ views), or (b) which provides only ($1/n$)-point increases at 1 correct replica per view thus requiring sustained correct behavior across multiple views even with non-negative score constraint. This creates a trade-off for attackers (mainly after GST): they must contribute positively for a period (or wait a considerable period) before gaining the opportunity to disrupt it - the overall performance improves compared to scenarios where they can disrupt without such constraints.

\begin{figure*}[h!]
	\centering
	\subfigure[\scriptsize Instantaneous throughput over first 2000 views under 1 Byzantine fault (Case 1).]{
		\label{fig:case11}
		\includegraphics[width=0.485\columnwidth]{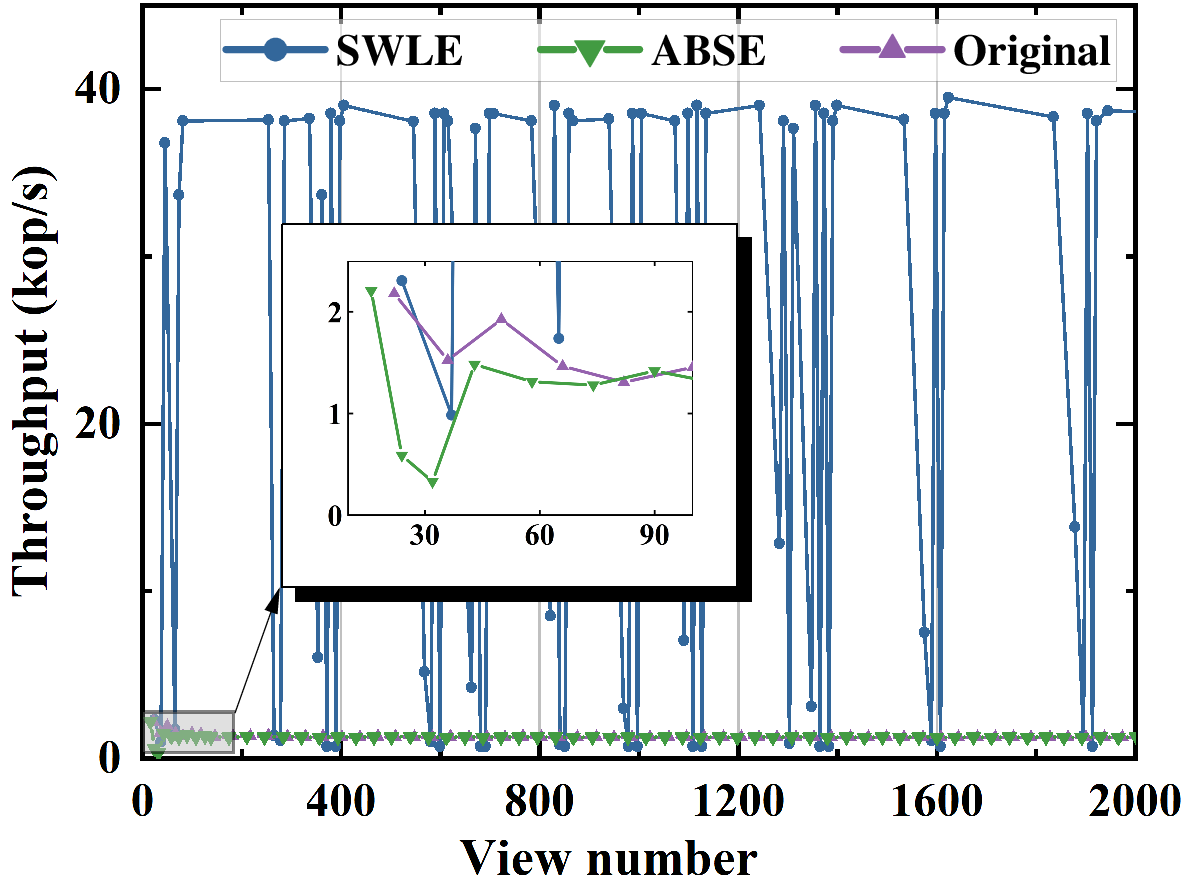}}
	\subfigure[\scriptsize Instantaneous latency over first 2000 views under 1 Byzantine fault (Case 1).]{
		\label{fig:case12}
		\includegraphics[width=0.485\columnwidth]{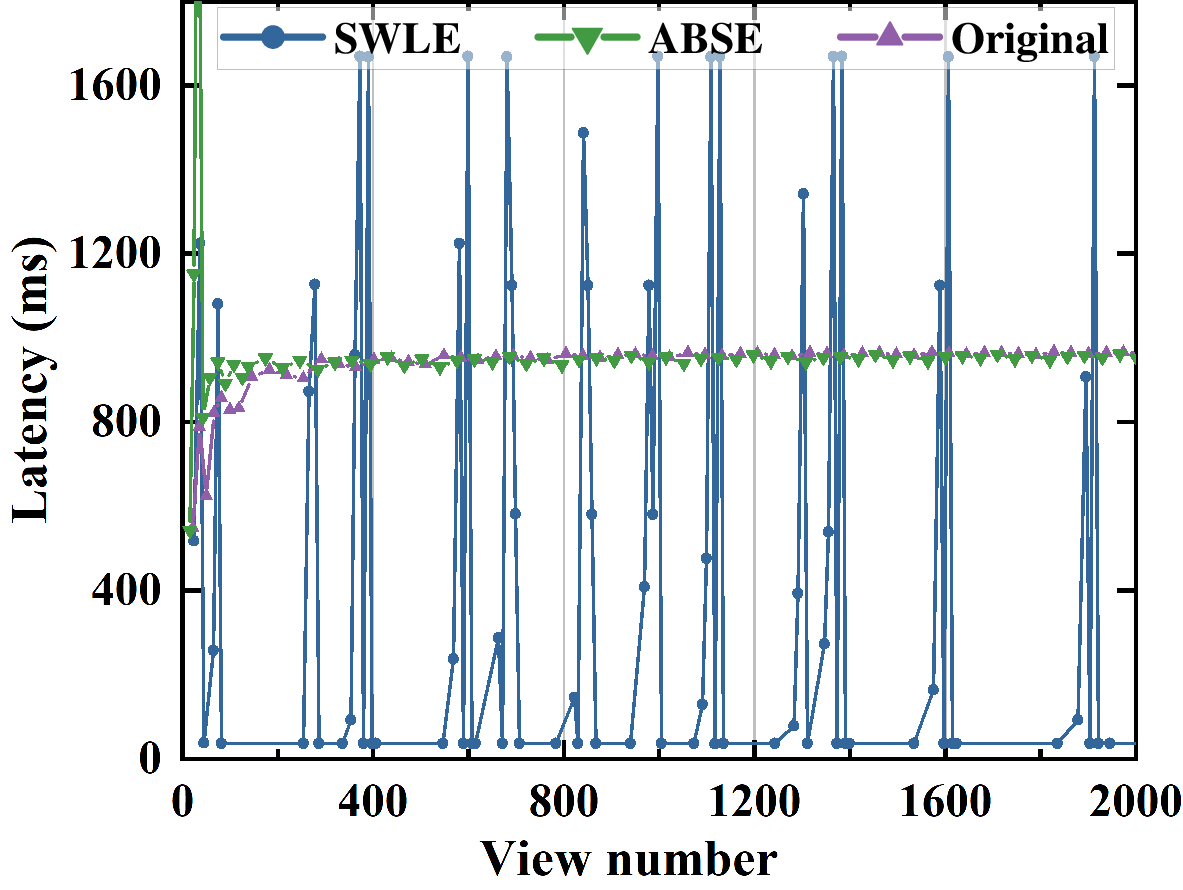}}
	\subfigure[\scriptsize Instantaneous throughput over first 2000 views under 3 Byzantine faults (Case 2).]{
		\label{fig:case21}
		\includegraphics[width=0.485\columnwidth]{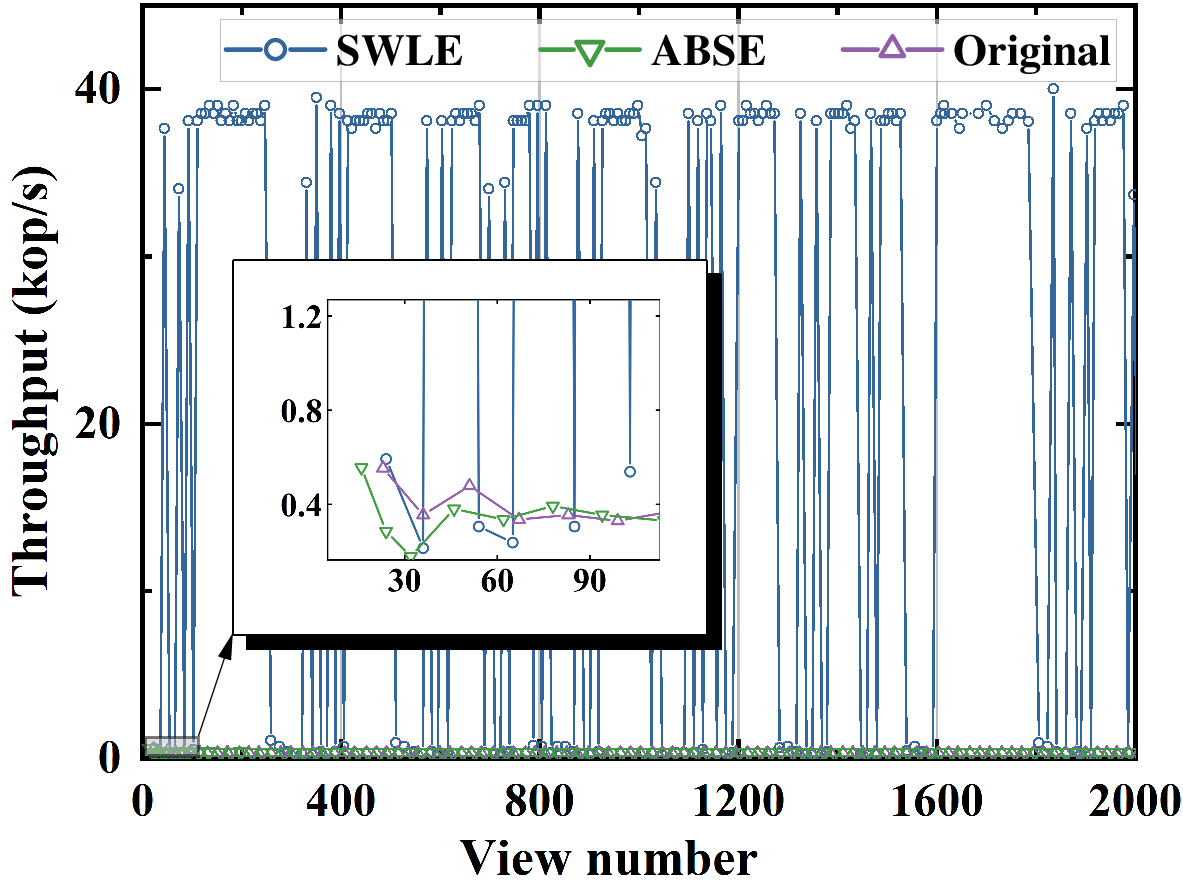}}
	\subfigure[\scriptsize Instantaneous latency over first 2000 views under 3 Byzantine faults (Case 2).]{
		\label{fig:case22}
		\includegraphics[width=0.485\columnwidth]{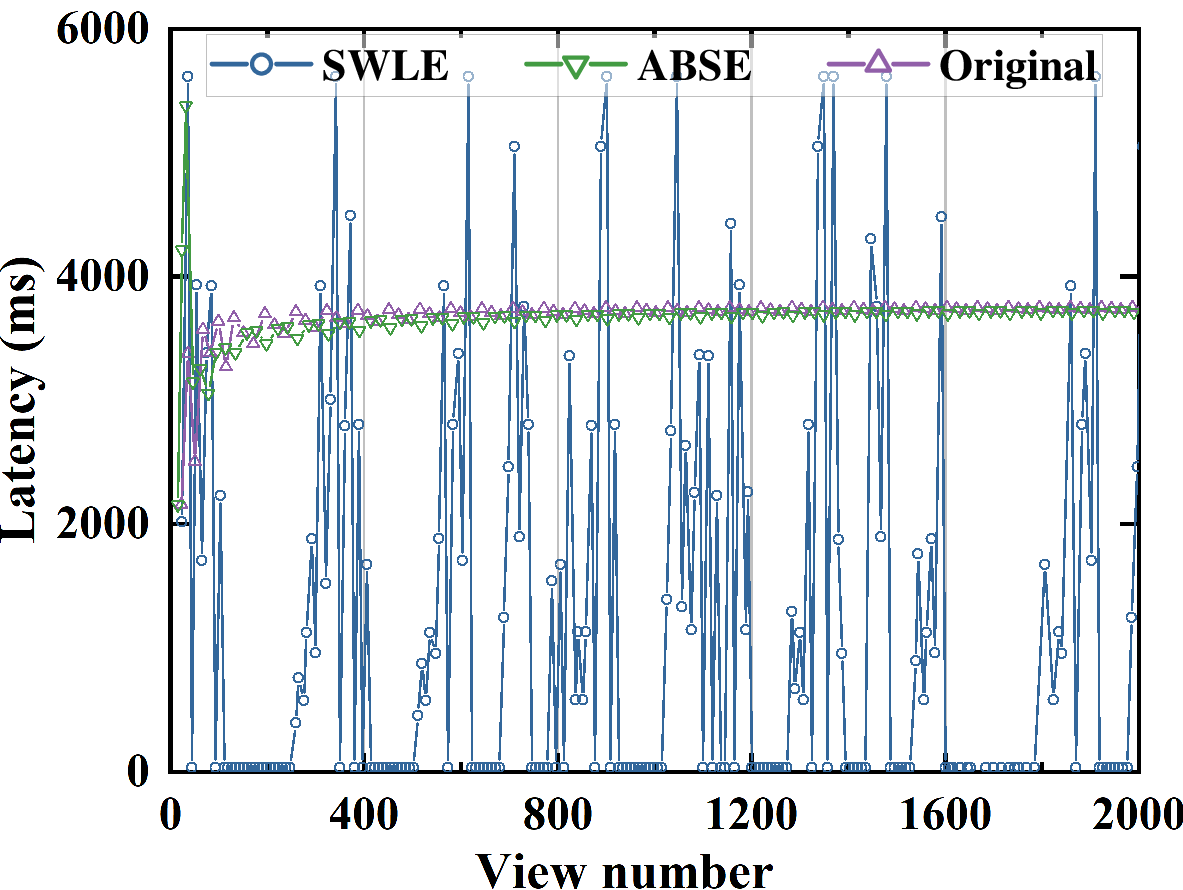}}
		\vspace{-0.05in}
	\subfigure[\scriptsize Instantaneous throughput over first 2000 views under 3 crash faults (Case 3).]{
		\label{fig:case31}
		\includegraphics[width=0.495\columnwidth]{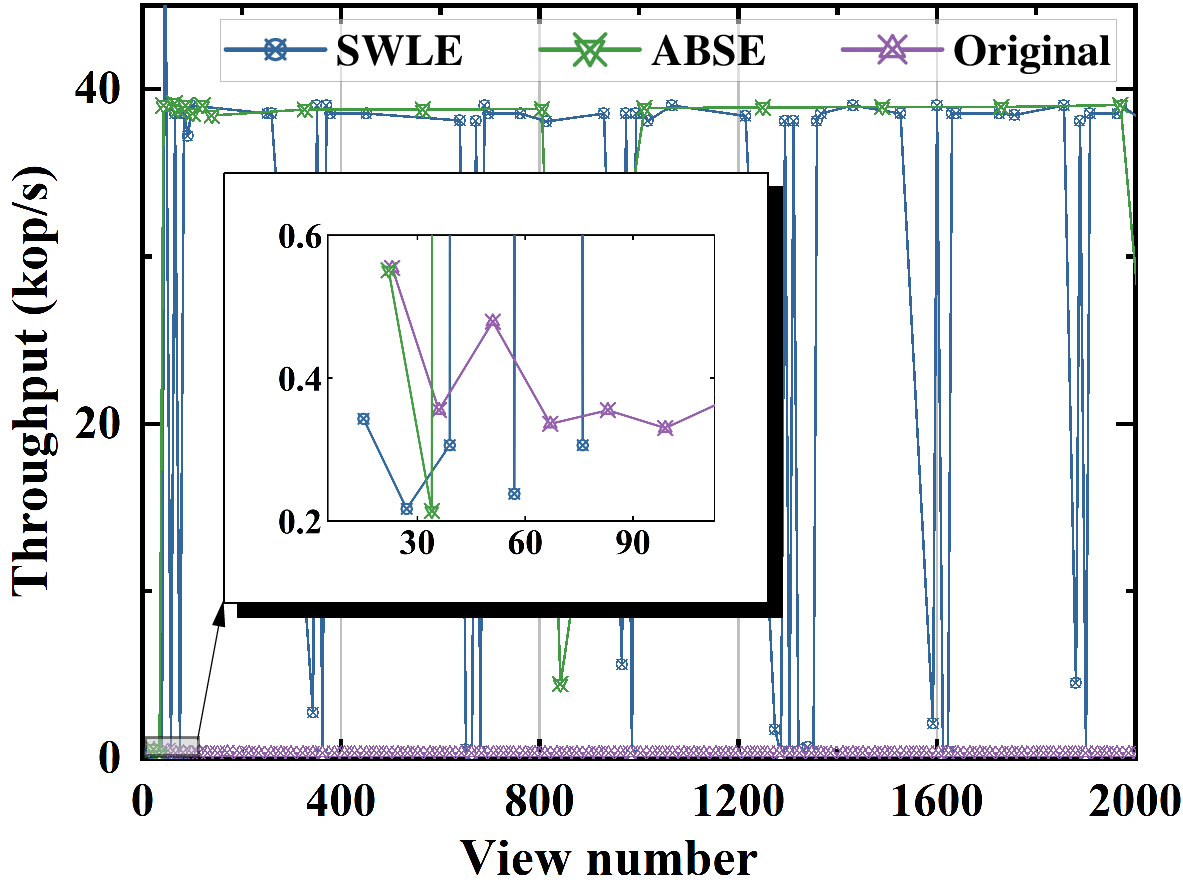}}
	\subfigure[\scriptsize Instantaneous latency over first 2000 views under 3 crash faults (Case 3).]{
		\label{fig:case32}
		\includegraphics[width=0.495\columnwidth]{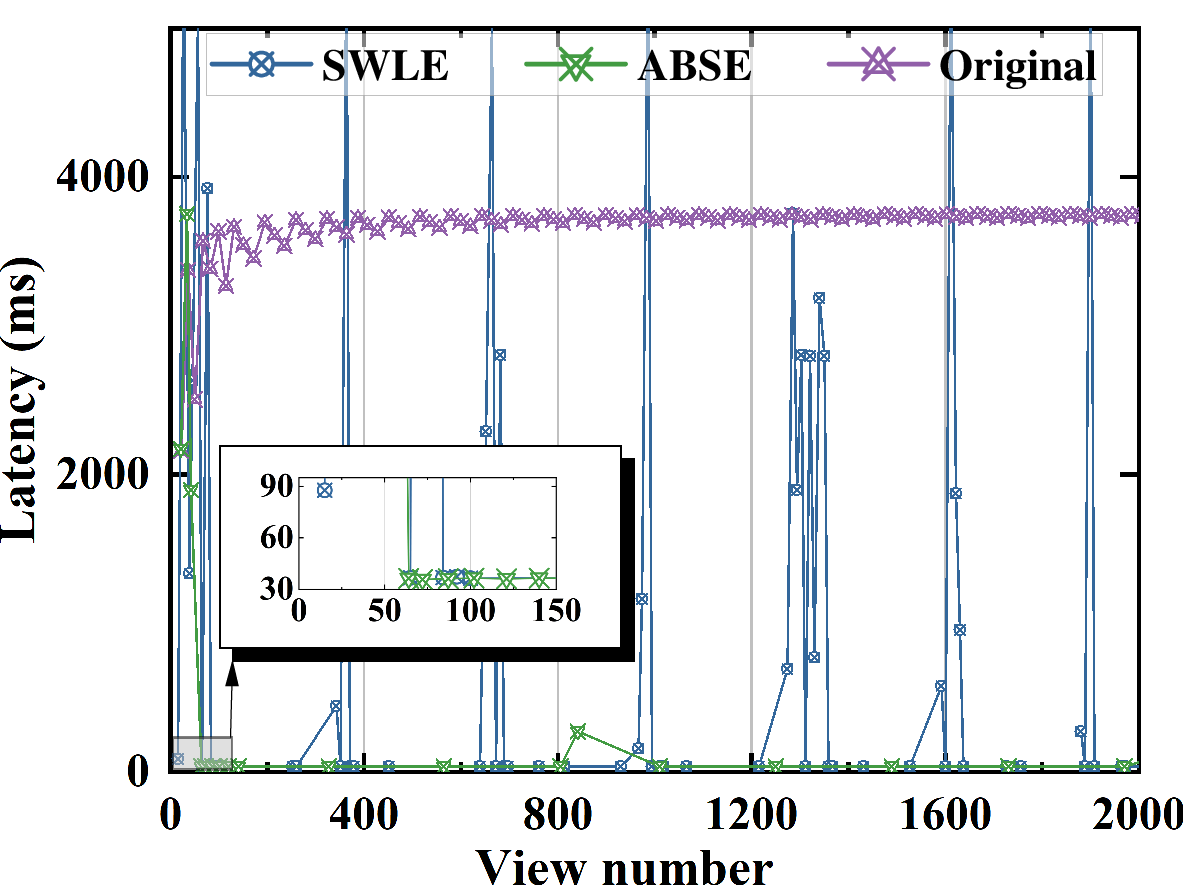}}
	\subfigure[\scriptsize Average throughput over the first 2000 views for the three cases.]{
		\label{fig:ath}
		\includegraphics[width=0.475\columnwidth]{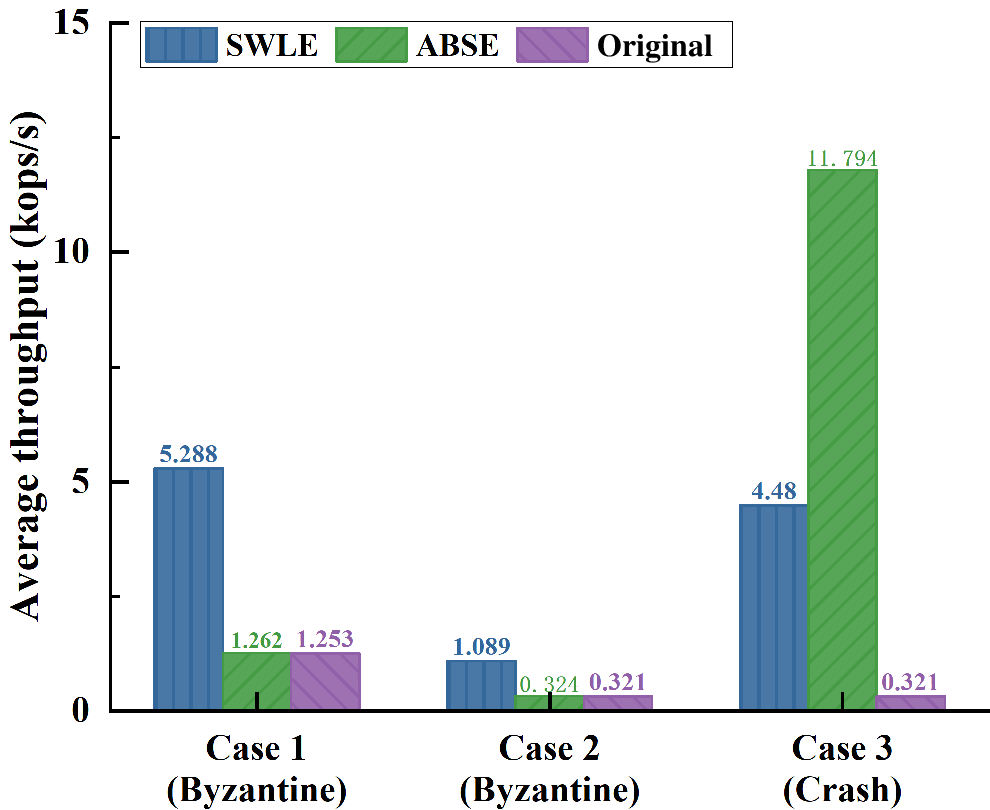}}
	\subfigure[\scriptsize Average latency over the first 2000 views for the three cases.]{
		\label{fig:al}
		\includegraphics[width=0.475\columnwidth]{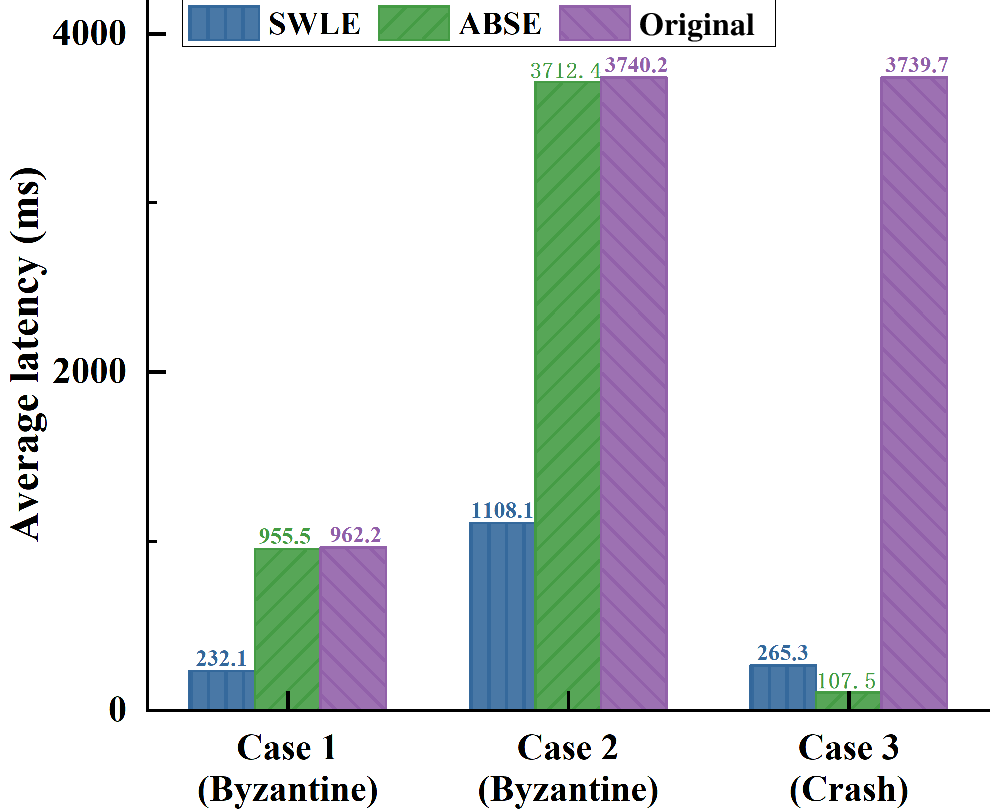}}
		\vspace{-0.03in}
	\subfigure[\scriptsize Percentage of views with faulty leaders (or timeout) over the first 2000 views.]{
		\label{fig:pmalic}
		\includegraphics[width=0.465\columnwidth]{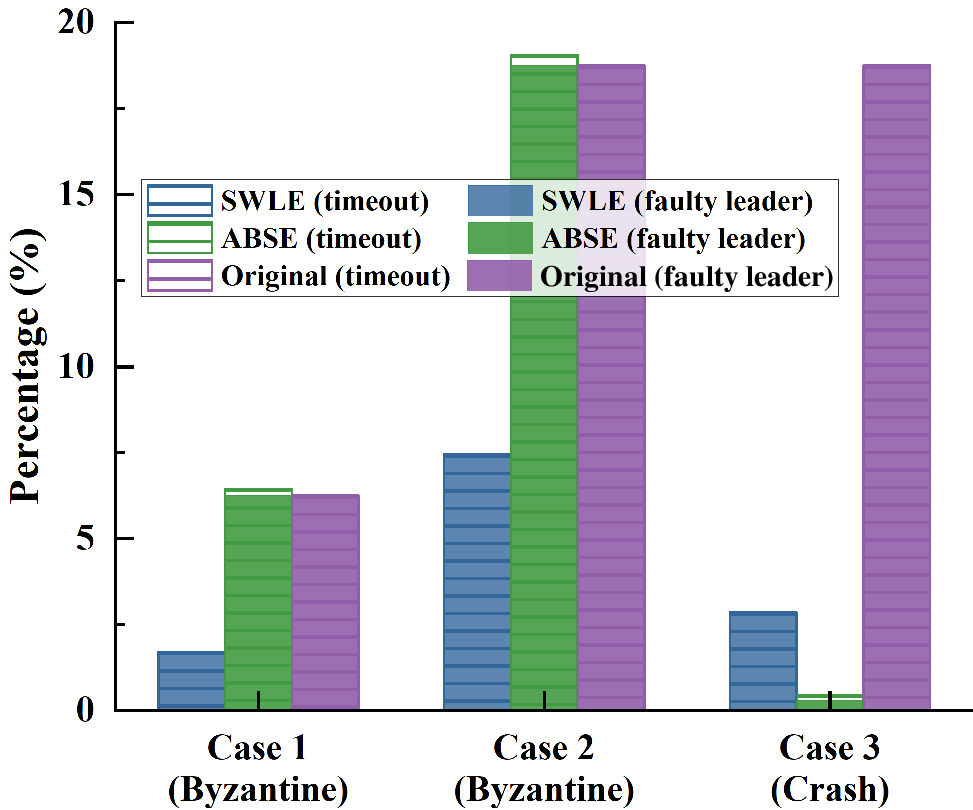}}
	\subfigure[\scriptsize Throughput-latency curves under different payload and batch sizes. 4 replicas in total (fault-free).]{
		\label{fig:benchmark}
		\includegraphics[width=\columnwidth]{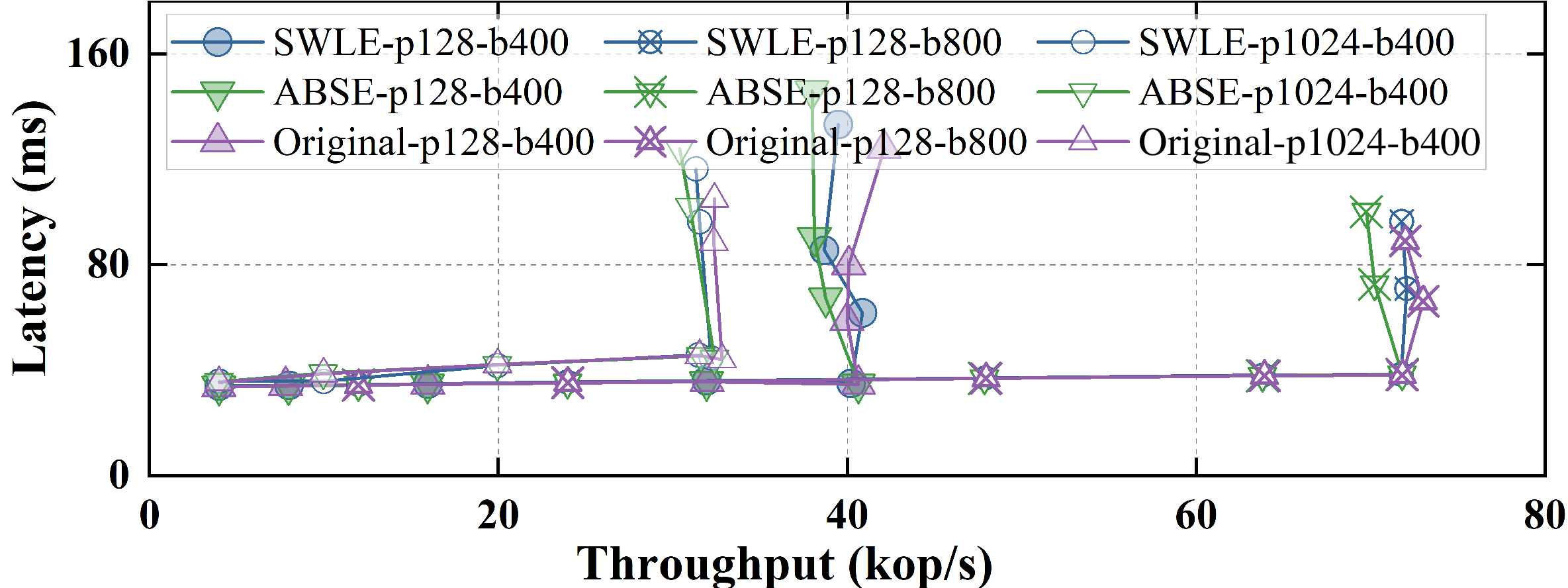}}
	\subfigure[\scriptsize Average throughput as the number of replicas increases (fault-free).]{
		\label{fig:scal}
	\includegraphics[width=0.495\columnwidth]{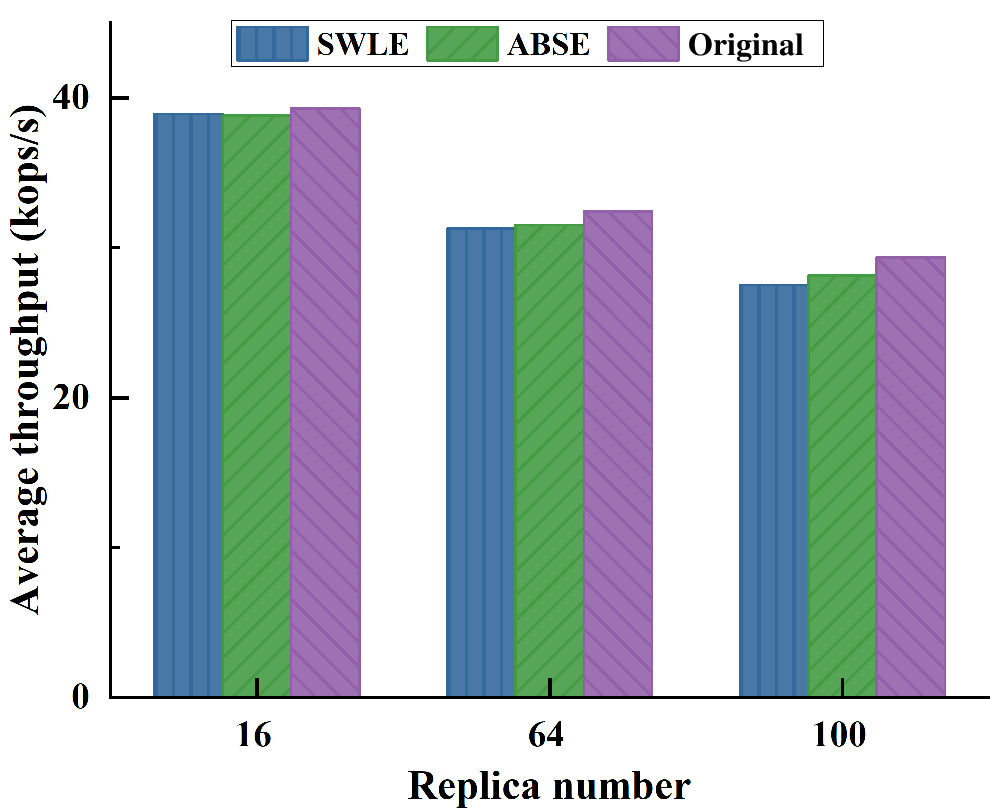}}
	\label{fig:performance}
	\caption{Performance.}
	\vspace{-0.13in}
\end{figure*}

\section{Evaluation}

\noindent \textbf{Overview.} We implemented SWLE as a protocol-independent Rust module containing all core functionalities and data structures (approximately 400 lines of code). The implementation relies on an extra cryptographic module (built primarily using ed25519\_dalek) alongside standard library dependencies. The two modules are open-sourced on GitHub\footnote[9]{https://github.com/BerserkRugal/SWLE} to ensure reproducibility of our experimental conclusions. 

For evaluation, we integrated SWLE into HotStuff protocol \cite{yin2019hotstuff} (non-chained version), requiring $\sim$300 lines of code modifications. We compare SWLE against ABSE, a state-of-the-art reputation-based leader election optimization framework (using the default open-source module provided in the original work \cite{liu2025abse}), and HotStuff's default round-robin leader election mechanism. For brevity, we refer to these three approaches as "SWLE", "ABSE", and "Original" respectively in subsequent figures. Our experiments demonstrate that SWLE: (1) imposes higher misbehavior costs on Byzantine replicas to maintain their leadership eligibility and provides superior performance under common Byzantine fault scenarios, (2) addresses critical limitations in ABSE where recovery time after GST remains unbounded, and (3) has a minor impact on base protocol performance in fault-free scenarios (lightweight nature).


\noindent \textbf{Setup.} We use throughput (operations finalized per second) and latency (duration from client operation request generation to finalization at replicas) as our main performance metrics. These metrics are categorized as: global average (average values calculated from system startup) and instantaneous metrics (sampled over sliding windows, with periodic reporting). The metric module is also open-sourced in our repository to ensure experimental reproducibility, but note it is not protocol-independent, as it depends on block metadata and configuration parameters (e.g., sliding window size).
 
We conducted experiments on 16 Alibaba Cloud Elastic Compute Service (ECS) u1-c1m2.4xlarge instances distributed across 4 different regions in northern China (Beijing, Zhangjiakou, Hohhot, and Ulanqab). Each machine provides 5Gbps bandwidth, 16 virtual CPUs on a 2.5GHz Intel Xeon Platinum scalable processor, 32GB memory, and runs Ubuntu 22.04 LTS. Replicas were distributed as equally as possible across these instances. Unless otherwise specified, the batch size 400 operations, operation size was 128 bytes, and the timeout was 1.5 seconds for each replica (note timeout configuration would influence observed performance under faults: shorter timeouts reduce degradation by faster leader rotation, while longer timeouts amplify performance gaps between mechanisms).

\noindent \textbf{Performance under Faulty Replicas.} 	We conducted experiments with a total of 16 replicas ($n$$=$$16$), prioritizing deployment of faulty replicas on machines with lower average message delays to maximize their advantage (and minimize their cost). Experiments were conducted under three fault scenarios: 1 Byzantine faulty replica (Case 1); 3 Byzantine faulty replica (Case 2); and 3 crash faulty replicas (Case 3). We note that experiments with 5 Byzantine/crash faulty replicas were also conducted but are omitted, as the performance variance significantly impacted graph readability, and the results showed similar trends to those observed under 1 or 3 faults. 

Byzantine faults strategically build reputation through rapid and valid voting (low-cost given their location advantage) then disrupt consensus (lead maliciously) when they gain leadership. Note in original HotStuff, no reputation mechanism exists, so Byzantine faults simply wait for their turn to become leaders without participating in voting. Crash faults discard all messages and remain unresponsive. Performance metrics were measured over the first 2000 views of system execution.

\vspace{1mm}

\noindent \textit{Case 1:} As shown in Figs. \ref{fig:case11}-\ref{fig:case12}, SWLE maintains throughput at $\sim$38 kop/s and latency at $\sim$37 ms in most views, with Byzantine leaders being elected and disrupting consensus infrequently. In contrast, ABSE and the original mechanism exhibit severe degradation: initial throughput briefly peaks at $\sim$2.2 kop/s ($\sim$600 ms latency) before subsequently stabilizing at $\sim$1.27 kop/s ($\sim$950 ms latency)\footnote[10]{It is observed that ABSE performs slightly better than original mechanism (e.g., throughput $\sim$1\% higher) after stabilization. This occurs because under ABSE, Byzantine replicas contribute valid votes to accumulate reputation scores, but this isn't the case under the original mechanism. Since they respond faster, consensus efficiency in rounds led by correct replicas is slightly higher.}. 
This persistent performance gap—visible across all sampling windows—indicates that Byzantine replicas can achieve stable election frequency, which suggest that both mechanisms provide weaker guarantees in leader election against common Byzantine behaviors compared to SWLE: as reflected in Fig. \ref{fig:pmalic}, SWLE experiences Byzantine leaders in <2\% of views, while the other two mechanisms exceed 6\%. Consequently, SWLE achieves  $\sim$5.3 kop/s average throughput (Fig. \ref{fig:ath})—$\sim$4.2× higher than baselines ($\sim$1.26 kop/s)—and $\sim$232 ms average latency (Fig. \ref{fig:al}), $\sim$75\% lower than $\sim$960 ms in ABSE/Original.

Notably, ABSE suffers an initial performance drop: Byzantine replicas exploit their faster responsiveness to dominate early voting, thereby preempting some correct replicas as consensus-promoting replicas per view. This leads to scenarios where no replicas can accumulate scores above the base reputation score thresholds (baseline) for leadership eligibility (score of Byzantine replicas are insufficient as well due to misleading) at $2f+1$ replicas during certain views in the initial period, resulting in ABSE's conflict scenarios (timeout occurs and roll back to the original HotStuff). This is also evidenced in Fig. \ref{fig:pmalic}, as ABSE's view timeout percentage exceeds the percentage of views led by Byzantine replicas. While recovery occurs quickly here, this exposes a critical limitation of ABSE: if adversaries consistently suppress correct replicas’ reputation accumulation at others pre-GST, recovery time becomes unbounded post-GST since the baseline scores increase with view counts in ABSE. SWLE avoids this by design, providing bounded post-GST recovery (which we have proved in theory).

\vspace{1mm}

\noindent \textit{Case 2 (similar to Case 1):} With increased Byzantine faults (Figs. \ref{fig:case21}–\ref{fig:case22}), ABSE/Original degrade further,  stabilizing at $\sim$0.33 kop/s throughput and $\sim$3.7 s latency. Though affected by Byzantine leaders in more views ($\sim$4×) compared to Case 1 (Fig. \ref{fig:pmalic}), SWLE still maintains relatively fast response for most operations. Its faulty-leader rate $\sim$7.5\% remains $\sim$60\% lower than ABSE/Original and sustains $\sim$1.1 kop/s average throughput—$\sim$3.3× higher than the other two—and $\sim$1.1 s latency (Fig. 1(g)–1(h)), $\sim$70\% lower than $\sim$3.7 s.

\vspace{1mm}

\noindent \textit{Case 3:} 
While SWLE isn’t specifically optimized for crash fault scenarios (Figs. \ref{fig:case31}–\ref{fig:case32}) and therefore does not outperform ABSE, it is still superior to Original: optimum throughput reaches 38 kop/s (similar to ABSE), but SWLE sustains this optimum for fewer views (the original mechanism performs similarly to Case 2). As shown in Fig. \ref{fig:pmalic}, crash replicas become leaders in $\sim$2.8\% of SWLE views—$\sim$9× more than ABSE—resulting in lower average throughput ($\sim$4.5 kop/s vs. ABSE’s $\sim$11.3 kop/s) and higher latency ($\sim$265 ms vs. $\sim$107 ms in Figs. \ref{fig:ath}-\ref{fig:al}).
However, optimizing specifically for crash fault scenarios in BFT systems may inadvertently create vulnerabilities exploitable by Byzantine adversaries.

\vspace{1mm}

\noindent \textbf{Base Performance (Fault-Free Case).} We measured throughput and latency of HotStuff under three mechanisms in a setting commonly used for evaluating other BFT systems \cite{yin2019hotstuff,bessani2014state} to test their performance characteristics (L-graphs). Experiments employed 4 replicas with two variable dimensions: (1) Operation (payload) sizes: 128 ("p128") and 1024 bytes ("p1024"), and (2) Batch sizes: 400 ("b400") and 800 ("b800"). We incrementally increased the client operation request rate until system saturation, maintaining identical rate steps (under the same configuration) across all mechanisms. As shown in Fig. \ref{fig:benchmark}, overall, all three mechanisms exhibit minimal peak throughput differences (e.g., $\sim$72 kop/s at p128-b800 and differences within 1\%). SWLE demonstrates slightly superior stability post-saturation compared to ABSE, though is slightly weaker than Original. 
Varying payload and batch sizes hardly cause changes in performance deviation between SWLE and Original.

\vspace{1mm}

\noindent \textbf{Scalability (Fault-Free Case).} We assess scalability by measuring average throughput at increasing replica counts ($n$=16, 64, 100). At $n$$=$$16$, SWLE achieves $\sim$38.8 kop/s — comparable to ABSE but slightly lower than Original. At $n=64$, SWLE begins to underperform compared to ABSE, achieving $\sim$31.3 kop/s — $\sim$1\% lower than ABSE and $\sim$4\% lower than Original. While these throughput gaps expand to $\sim$2.5\% and $\sim$7\% respectively when $n=100$. 
The widening gap stems from SWLE’s per-view overhead in candidate generation and leader determination across all replicas,
making the mechanism somewhat sensitive to increases in replica numbers.
This highlights a potential optimization direction for future work.

\section{Conclusions}

We address limitations in existing reputation-based leader election frameworks for partially synchronous BFTs. We present a novel protocol-independent abstraction formalizing three core properties for theoretical analysis and design. Building on this, we design SWLE, a novel reputation-based leader election mechanism that provides enhanced guarantees. We show, with a up to  16-server deployment, SWLE achieves superior performance to the state-of-the-art solution under common Byzantine faults, while maintaining efficiency in fault-free scenarios.

\bibliographystyle{IEEEtran}
\bibliography{myrefs}

\end{document}